\documentclass{SciPost}

\usepackage{graphicx}
\usepackage{amsmath}
\usepackage{amssymb}
\usepackage{amsthm}
\usepackage{mathrsfs}
\usepackage{extarrows} 
\usepackage{subcaption}
\usepackage{xcolor}
\usepackage{tikz}
\usepackage{braket}
\usepackage{comment}
\usepackage{bbm}
\usetikzlibrary{math}
\usetikzlibrary{decorations.pathmorphing}
\usepackage{enumitem}

\usepackage{breakurl}
\hypersetup{
    colorlinks,
    linkcolor={black},
    citecolor={blue!50!black},
    urlcolor={blue!80!black}
}

\theoremstyle{plain}
\newtheorem{thm}{Theorem}[section]

\newtheorem{lem}[thm]{Lemma}

\theoremstyle{definition}

\newcommand{\wt}{\widetilde}
\newcommand{\Om}{\Omega}
\newcommand{\Ga}{\Gamma}
\newcommand{\La}{\Lambda}
\newcommand{\Sig}{\Sigma}
\newcommand{\hp}{\hphantom}

\begin{document}

\begin{center}{\Large \textbf{
Anomaly diagnosis via symmetry restriction in two-dimensional lattice systems\\
}}\end{center}

\begin{center}
Kyle Kawagoe\textsuperscript{1} and
Wilbur Shirley\textsuperscript{2,3}
\end{center}

\begin{center}
{\bf 1} Department of Physics, Department of Mathematics, and Center for Quantum Information Science and Engineering, The Ohio State University, Columbus, OH
\\
{\bf 2} Leinweber Institute for Theoretical Physics, University of Chicago, Chicago, IL
\\
{\bf 3} School of Natural Sciences, Institute for Advanced Study, Princeton, NJ
\\
\end{center}

\begin{center}
\today
\end{center}

\section*{Abstract}
{
We describe a method for computing the anomaly of any finite unitary symmetry group $G$ acting by finite-depth quantum circuits on a two-dimensional lattice system. 
The anomaly is characterized by an index valued in the cohomology group $H^4(G,U(1))$, which generalizes the Else-Nayak index for locality preserving symmetries of quantum spin chains. We show that a nontrivial index precludes the existence of a trivially gapped symmetric Hamiltonian; it is also an obstruction to ``onsiteability" of the symmetry action.
We illustrate our method via a simple example with $G=\mathbb{Z}_2\times\mathbb{Z}_2\times\mathbb{Z}_2\times\mathbb{Z}_2$. Finally, we provide a diagrammatic interpretation of the anomaly formula which hints at a higher categorical structure.
}

\vspace{10pt}
\noindent\rule{\textwidth}{1pt}
\tableofcontents
\noindent\rule{\textwidth}{1pt}
\vspace{10pt}

\section {Introduction}

A symmetry of a quantum many-body system is \textit{anomalous} if it constrains the low energy dynamics of the system. Specifically, an anomaly precludes the possibility of a symmetric, gapped, local Hamiltonian with a unique invertible ground state \cite{Lieb,Affleck,Oshikawa,Hastings,Ogata1,Ogata2,KapustinSopenko}. A system with an anomalous symmetry must instead have one (or more) of the following nontrivial features: gapless modes, spontaneous symmetry breaking, or fractionalized excitations. Such anomalous symmetries arise naturally in the low energy theories describing spatial boundaries of symmetry-protected topological (SPT) phases \cite{CallanHarvey,Ludwig,LevinGu,ChenSPT,Senthil,Wen}.

Though they bear dynamical consequences, anomalies themselves are a purely kinematic property of the symmetry operators. For instance, in zero spatial dimensions (0D),\footnote{Throughout, we denote $d$ spatial dimensions by $d$D.} a unitary symmetry is anomalous when it acts on the Hilbert space as a nontrivial projective representation of the symmetry group $G$. This correspondence allows for a simple classification of 0D anomalies in terms of the cohomology group $H^2(G,U(1))$ \cite{Chen1D,Schuch1D,Pollmann1D}.\footnote{We note that 0D anomalies are in one-to-one correspondence with 1D SPTs.} In higher dimensions, we advocate that an anomaly should be viewed as a form of entanglement borne by the symmetry operators. For unitary, internal symmetries in 1D and 2D bosonic systems, this form of entanglement is characterized by indices valued in the cohomology groups $H^3(G,U(1))$ and $H^4(G,U(1))$, respectively \cite{KapustinThorngren,ElseNayak}. This scheme mirrors the classification of SPT phases in one higher dimension \cite{ChenSPT}.

In general, however, it is not clear how to define these anomaly indices in specific systems. It is thus an important challenge to develop systematic methods to identify the anomaly class of a symmetry given its microscopic form. In 1D, this problem has been studied from several vantage points. Chen, Liu, and Wen first explained how to compute the anomaly index $[\alpha]\in H^3(G,U(1))$ in the special case of symmetries represented by matrix product unitary operators \cite{CZX}. Later, Else and Nayak presented a method applicable to any symmetry acting by finite depth quantum circuits (FDQC), based on the idea of spatially restricting the symmetry action to a finite interval \cite{ElseNayak}. More recent works have explored methods that identify the anomaly 3-cocycle with the $F$-symbols characterizing fusion of domain walls\cite{KawagoeLevin} and symmetry defects\cite{Sahand}. Furthermore, approaches to computing anomalies in conformally invariant SPT boundary theories have been studied in \cite{CFT1,CFT2,CFT3,CFT4,CFT5}.

In contrast, such methods in 2D are less well understood. The first inroad was made by Else and Nayak, who described a solution in the special case of symmetries with a ``nearly on-site" form \cite{ElseNayak}.\footnote{We call a symmetry ``nearly on-site" if it is of the form $U_g=N_gS_g$ where $S_g=\sum_\alpha\ket{g\alpha}\bra{a}$ for some on-site symmetry action $\alpha\to g\alpha$ on the classical label $\alpha$, and $N_g=\sum_\alpha e^{i\mathcal{N}_g[\alpha]}\ket{\alpha}\bra{\alpha}$ is a phase factor determined by a local functional $\mathcal{N}_g$ of the configuration $\alpha$.} More recently, Kobayashi \textit{et. al.} have shown how to compute anomaly invariants for abelian symmetry groups by studying the algebraic structure of symmetry defect loops \cite{Kobayashi}. Although these works offer valuable insights, they do not address the problem in full generality. The purpose of our work is to provide a systematic method for computing the anomaly index $[\omega]\in H^4(G,U(1))$ of any unitary symmetry acting by FDQCs on a 2D lattice system composed of qudits/spins. We argue that a nontrivial $H^4(G,U(1))$ index constitutes an obstruction to a local symmetric Hamiltonian with a unique, invertible, gapped ground state. Hence we conclude that the index we define correctly labels the anomaly. Our method can in turn be used to identify the SPT order of a bulk 3D system, whenever the 2D boundary theory has a tensor product Hilbert space.

Our method is a generalization of the Else-Nayak method for 1D systems, and reduces to their 2D solution in the case of ``nearly on-site" symmetries. Underlying the calculation is the essential idea that an anomaly is an obstruction to spatially restricting the symmetry to a finite region, such that the restricted symmetries themselves form a representation of the symmetry group. Intuitively, this obstruction arises from the nontrivial entanglement carried by anomalous symmetry operators, which prevents the symmetry action in one region from being decoupled from the action in the region's complement. In other words, the anomaly is an obstruction to ``onsiteability" of the symmetry action.\footnote{A precise definition of the notion of ``onsiteability" is given in Sec.~\ref{sec:discussion}.} In 1D it has been shown that the anomaly index $[\alpha]\in H^3(G,U(1))$ is the complete obstruction to onsiteability \cite{SeifnashriShirley}. In contrast, in 2D the anomaly index $[\omega]\in H^4(G,U(1))$ is not the only such obstruction---there is an additional obstruction characterized by an index $[\nu]\in H^2(G,\mathbb{Q}_+)$ where $\mathbb{Q}_+$ is the group of GNVW indices \cite{GNVW} of 1D quantum cellular automata (QCA) \cite{Kapustin,H2}. Our procedure not only identifies the $H^4$ anomaly index, but also computes the $H^2$ index as a byproduct of the anomaly computation. When the $H^2$ index is nontrivial, an additional step must be taken in the procedure in order to ``cancel" the $H^2$ index before proceeding with the calculation. We emphasize that, while a nontrivial $H^4$ index indicates both an obstruction to onsiteability and an anomaly in the sense of a constraint on the low energy dynamics, a nontrivial $H^2$ index is an obstruction to onsiteability but is not indicative of an anomaly.\footnote{This point is elucidated in \cite{H2}, which explicitly constructs trivially gapped Hamiltonians whose symmetries have nontrivial $H^2(G,\mathbb{Q}_+)$ indices.}

The paper is organized as follows. In Sec.~\ref{sec:background}, we review basic notions of QCAs, FDQCs, and the GNVW index. In Sec.~\ref{sec:index} we present our procedure for anomaly diagnosis in 2D systems. In Sec.~\ref{sec:obstruction}, we argue that a 2D $G$-symmetry with a nontrivial $H^4(G,U(1))$ index is incompatible with a symmetric local Hamiltonian with a unique, invertible, gapped ground state. In Sec.~\ref{sec:example}, we discuss a simple example with $G=\mathbb{Z}_2\times\mathbb{Z}_2\times\mathbb{Z}_2\times\mathbb{Z}_2$. In Sec.~\ref{diagrams}, we provide a diagrammatic interpretation of the anomaly formula. In Sec.~\ref{sec:discussion} we conclude with an overview of future directions.

\section {Preliminaries}
\label{sec:background}

In this work, we consider systems defined on a finite 2D lattice $\Lambda$. Our interest is in physical properties that can be deduced from studying a finite, local patch of the system. For our purposes boundary conditions will thus play no role. Throughout, we consider a tensor product Hilbert space $\mathcal{H}=\bigotimes_{i\in\Lambda}\mathcal{H}_i$ where $\mathcal{H}_i$ is the local Hilbert space on site $i$. We assume that our system does not contain any fermionic degrees of freedom.

In this setting, a unitary operator that strictly preserves locality is referred to as a \textit{quantum cellular automaton} (QCA) \cite{Farrelly}.\footnote{In infinite volume systems, QCAs are defined as bounded spread $*$-isomorphisms of the local operator algebra. In finite systems, we may identify a QCA with the adjoint action of a unitary operator.} Specifically, a QCA is a unitary operator $U$ with the following property: there exists some $R>0$ such that, for every site $i\in\Lambda$ and every operator $O_i$ supported on $i$, the transformed operator $UO_iU^{-1}$ is supported in a disk of radius $R$ centered around $i$. The number $R$ is referred to as the \textit{range} of $U$; in general $R$ is much smaller than the system size.

A \textit{finite depth quantum circuit} (FDQC) is a unitary operator that can be represented as a quantum circuit consisting of a finite number of layers of non-overlapping quantum gates of uniformly bounded diameter.\footnote{A quantum gate is a unitary operator supported on a finite number of sites. The diameter of a gate is the diameter of the smallest ball containing the support of the gate.} The number of layers in a FDQC is referred to as its depth $D$. A FDQC of gate diameter $k$ and depth $D$ is manifestly a QCA with range $R<kD$. We note that the representation of a FDQC as a circuit is not canonical.

Although every FDQC is a QCA, the converse does not hold: Not every QCA can be realized as a FDQC \cite{GNVW,FHH}. A QCA with this property is regarded as \textit{nontrivial}, and the set of QCAs modulo FDQCs forms an abelian group under composition \cite{HastingsFreedman}.\footnote{Strictly speaking, when classifying QCAs one should allow for stabilization, \textit{i.e.} identifying a QCA $U$ with $U\otimes\mathbbm{1}$, where $\mathbbm{1}$ acts on an ancillary system.} In 1D, the group of nontrivial QCAs is isomorphic to $\mathbb{Q}_+$, the multiplicative group of positive rational numbers \cite{GNVW,HastingsFreedman}. Accordingly, every 1D QCA is characterized by a rational number $\nu$ known as the GNVW index. Roughly speaking, this index quantifies the operator ``flow" along the chain. A simple example of a 1D QCA with GNVW index equal to an integer $n$ is given by a uniform translation of an $n$-dimensional qudit chain by one site to the right. Conversely, a uniform translation by one site to the left has GNVW index $1/n$.

Given a QCA $U$ of range $R$ and a region $A\subset\Lambda$, a \textit{restriction} of $U$ to $A$ is a QCA $U^A$ which acts like $U$ deep within $A$, and acts like the identity far outside $A$. To make this notion precise, we need to specify what we mean by ``deep within" and ``far outside" the region $A$. This is somewhat arbitrary---to be concrete, we define the region $\text{Int}(A)\subset A$ as the set of sites in $A$ whose distance from $\partial A$ (the boundary of $A$) is at least, say, $5R$. Similarly, we define the region $\text{Ext}(A)\subset\overline{A}$ as the set of sites outside $A$ whose distance from $\partial A$ is at least $5R$. We then require a restriction $U^A$ to satisfy the following property: For any operator $O_i$ supported on a site $i$ in $\text{Int}(A)$ or $\text{Ext}(A)$,
\begin{equation}
    U^AO_i(U^A)^{-1}=\begin{cases}
    UO_iU^{-1} & i\in\text{Int}(A)\\
    O_i & i\in\text{Ext}(A).
\end{cases}
\end{equation}
By definition, $U^A$ is supported in $\overline{\text{Ext}(A)}$.

The following is an important fact: A QCA is trivial, \textit{i.e.} it is a FDQC, if and only if it admits restrictions to arbitrary finite regions \cite{HastingsFreedman}. More precisely, for any nontrivial QCA $V$ and sufficiently large finite region $A$, there does not exist a restriction of $V$ to $A$; on the other hand, for any FDQC $U$ and region $A$, there does exist a restriction $U^A$ of $U$ to $A$.\footnote{To see this, consider for instance a FDQC composed of all gates of $U$ that are strictly supported in $A$.} However, the restriction $U^A$ is not unique. Indeed, given any 1D QCA $\Sigma$ supported near $\partial A$, the operator $\Sigma U^A$ is also a valid restriction of $U$ to $A$.\footnote{This requires that $\Sigma$ is supported on a thin strip of width $10R$ centered around $\partial A$. Note that $\Sigma$ is a 2D FDQC even when it is nontrivial as a 1D QCA.}
Moreover, there is no canonical restriction of $U$---in general, a choice of restriction of $U$ is inherently arbitrary.

\section {The anomaly index}
\label{sec:index}

In this section, we explain our procedure for computing the anomaly of a finite group $G$-symmetry in a 2D lattice system. Before delving into a precise exposition of the procedure, we provide a brief overview as a guide for the reader.

\subsection{Overview}

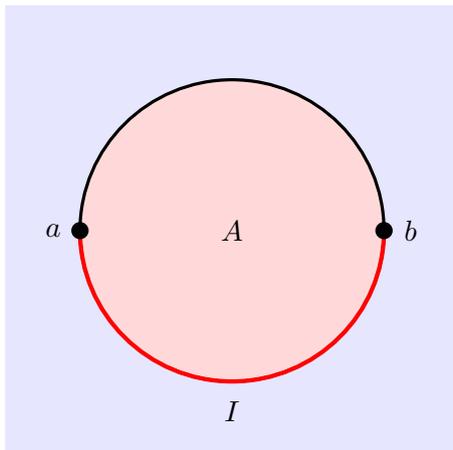
\begin{figure}
    \centering
    \begin{tikzpicture}
        \filldraw[thick,white,fill=blue!10](-1,-3)--(-1,3)--(5,3)--(5,-3)--(-1,-3);
        \filldraw[very thick,black,fill=red!15](2,0) circle(2);
        \draw[ultra thick,red](4,0) arc(0:-180:2);
        \node at (2,0){$A$};
        \node at (2,-2.4){$I$};
        \filldraw[thick,black](0,0) circle(.1);
        \filldraw[thick,black](4,0) circle(.1);
        \node at (-.35,0){$a$};
        \node at (4.35,0){$b$};
    \end{tikzpicture}
    \caption{Spatial arrangement of the disk $A$, interval $I\in\partial A$, and its endpoints $a$ and $b$.}
    \label{fig:regions}
\end{figure}

We consider a tensor product Hilbert space $\mathcal{H}=\bigotimes_{i\in\Lambda}\mathcal{H}_i$ on a 2D lattice $\Lambda$. The input to the procedure is a collection $\{U_g\}_{g\in G}$ of FDQCs satisfying the group multiplication law $U_gU_h\propto U_{gh}$ (possibly up to a phase).\footnote{We do not assume that the symmetry operators are translation invariant, nor do we assume that the lattice itself has any spatial symmetries.} The output consists of a 2-cocycle $\nu:G\times G\to\mathbb{Q}_+$, and a 4-cocycle $\omega:G\times G\times G\times G\to U(1)$. At various steps, the procedure involves making certain arbitrary choices, and the cocycles that are obtained depend on these choices. Crucially, the resulting ambiguity in $\nu$ is precisely a 2-coboundary, and the ambiguity in $\omega$ is precisely a 4-coboundary. Thus, the procedure yields well-defined cohomology classes $[\nu]\in H^2(G,\mathbb{Q}_+)$ and $[\omega]\in H^4(G,U(1))$. As discussed in the introduction, $[\omega]$ labels the anomaly of the symmetry in the sense of a constraint on low energy dynamics, whereas $[\nu]$ is independent of such dynamical constraints. Both indices, when nontrivial, represent obstructions to onsiteability of the symmetry.

Our procedure consists of the following set of concrete steps (see below for an explanation of notation conventions):
\begin{enumerate}[leftmargin=*]
    \item Choose a large disk $A$ and an interval $I=[a,b]\subset\partial A$, as depicted in Fig.~\ref{fig:regions}.

    \item Choose a restriction $U^A_g$ of $U_g$ to $A$, for all $g\in G$.
    
    \item Define the operator $\Om_{g,h}\equiv U^A_gU^A_h(U^A_{gh})^{-1}$ for all $g,h\in G$. By definition $\Om_{g,h}$ is a 1D QCA supported on a thin strip along $\partial A$. Let $\nu(g,h)$ be the GNVW index of $\Om_{g,h}$. The function $\nu:G\times G\to\mathbb{Q}_+$ is a 2-cocycle, \textit{i.e.} it satisfies condition (\ref{eq:2cocycle}).
    
    \item Introduce an ancillary Hilbert space $\mathcal{H}_{\partial A}$ describing a 1D chain of qudits living along $\partial A$, and define the enlarged Hilbert space $\mathcal{H}'=\mathcal{H}\otimes\mathcal{H}_{\partial A}$. Define a canonical 1D QCA $T_{g,h}$ on $\mathcal{H}_{\partial A}$ with GNVW index $\nu(g,h)^{-1}$. Finally let $\Om'_{g,h}=\Om_{g,h}\otimes T_{g,h}$, which is a 1D FDQC.
    
    \item Choose a restriction $\Om^I_{g,h}$ of $\Om'_{g,h}$ to the vicinity of the interval $I$, for all $g,h\in G$.
    
    \item Define the operator
    \begin{equation*}
        \Gamma_{g,h,k}\equiv\Om^I_{g,h}\Om^I_{gh,k}\left(\hp{}^g\Om^I_{h,k}\Om^I_{g,hk}\right)^{-1}
    \end{equation*}
    for all $g,h,k\in G$. By definition $\Gamma_{g,h,k}$ is a 0D QCA supported near points $a$ and $b$.
    
    \item Choose a restriction $\Ga^a_{g,h,k}$ of $\Ga_{g,h,k}$ to point $a$, for all $g,h,k\in G$.
    
    \item Decompose $I$ into two subintervals $L$ and $R$, as depicted in Fig.~\ref{fig:LR}. Choose an arbitrary decomposition $\Om^I_{g,h}=\Om^L_{g,h}\Om^R_{g,h}$ where $\Om^L_{g,h}$ ($\Om^R_{g,h}$) are 1D FDQCs supported near $L$ ($R$). Define the operator
    \begin{equation*}
        \Delta^a_{g,h,(k,l)}\equiv\hp{}^{(g,h)\cdot gh}\Om^L_{k,l}\hp{}\left(\hp{}^{g\cdot h}\Om^L_{k,l}\right)^{-1}
    \end{equation*}
    for all $g,h,k,l\in G$. By definition, $\Delta^a_{g,h,(k,l)}$ is supported near point $a$.
    
    \item Finally, define the function $\omega:G\times G\times G\times G\to U(1)$ as follows:
    \begin{equation*}
        \omega(g,h,k,l)\equiv\Ga^a_{g,h,k}\cdot\hp{}^{^g(h,k)}\Ga^a_{g,hk,l}\cdot\hp{}^g\Ga^a_{h,k,l}\left(\hp{}^{(g,h)}\Ga^a_{gh,k,l}\cdot\hp{}\Delta^a_{g,h,(k,l)}\cdot\hp{}^{^{g\cdot h}(k,l)}\Ga^a_{g,h,kl}\right)^{-1}.
    \end{equation*}
    By construction, $\omega$ is a 4-cocycle, \textit{i.e.} it satisfies condition (\ref{eq:4cocycle}).
\end{enumerate}

\begin{figure}
    \centering
    \begin{tikzpicture}
        \draw[ultra thick,red](4,0) arc(0:-180:2);
        \node at (2,-2.4){$c$};
        \node at (.35,-1.7){$L$};
        \node at (3.65,-1.7){$R$};
        \filldraw[thick,black](0,0) circle(.1);
        \filldraw[thick,black](4,0) circle(.1);
        \filldraw[thick,black](2,-2) circle(.1);
        \node at (-.35,0){$a$};
        \node at (4.35,0){$b$};
    \end{tikzpicture}
    \caption{Decomposition of the interval $I$ into two subintervals $L=[a,c]$ and $R=[c,b]$.}
    \label{fig:LR}
\end{figure}
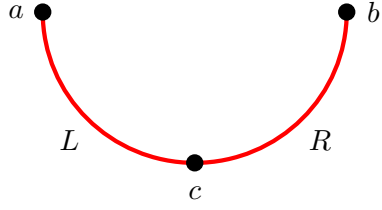

\noindent\textbf{Notation}: We have adopted the following notational conventions. First, $^VW\equiv VWV^{-1}$ for any pair of unitary operators $V$ and $W$. Moreover, $U^A_g$ is denoted by $g$ and $\Om^I_{g,h}$ by $(g,h)$ whenever they appear within a left superscript. For instance,
\begin{align}
    ^gW&\equiv U^A_gW(U^A_g)^{-1},\\^{g\cdot h}W&\equiv U^A_gU^A_hW(U^A_gU^A_h)^{-1},\\\hp{}^{(g,h)}W&\equiv\Om^I_{g,h}W(\Om^I_{g,h})^{-1},\\\hp{}^{^g(h,k)}W&\equiv U^A_g\Om^I_{h,k}(U^A_g)^{-1}W\big(U^A_g\Om^I_{h,k}(U^A_g)^{-1}\big)^{-1}.
\end{align}

\noindent\textbf{Remark:} It is possible to choose operator restrictions such that the output 4-cocycle $\omega$ is automatically ``normalized", \textit{i.e.} $\omega(g,h,k,l)=1$ when any member of the 4-tuple $g,h,k,l$ is the identity element. To ensure that $\omega$ is normalized, one must choose $U^A_1$ to be the identity operator, $\Om^I_{g,h}$ to be the identity operator whenever $g=1$ or $h=1$, and $\Ga^a_{g,h,k}$ to be the identity operator when any of $g,h,k$ is the identity.

\subsection{Definition of the index}
\label{sec:explicit}

We now give a detailed exposition of the procedure following the list of steps in the overview above.
\begin{enumerate}[leftmargin=*]
    \item To begin, we choose a disk $A$ and an interval $I=[a,b]\subset\partial A$. The radius of $A$, and the length of $I$, are assumed to be much larger than $R$, where $R$ is the maximal range of the symmetry operators $\{U_g\}$.
    
    \item Next, we choose a restriction $U^A_g$ of each symmetry operator $U_g$ to the disk $A$.

    \item The restricted symmetries obey the group multiplication law \textit{up to} composition by a 1D QCA supported near $\partial A$. In particular, for each pair $g,h,\in G$, define the 1D QCA
    \begin{equation}
        \Om_{g,h}\equiv U^A_g U^A_h\left(U^A_{gh}\right)^{-1}.
    \end{equation}
    For every triple $g,h,k\in G$, the $\Om$ operators obey a constraint which we refer to as the ``non-abelian 2-cocycle condition" \cite{brown1982cohomology,Muger,ElseNayak}:
    \begin{equation}
        \label{eq:NA2cocycle}
        \Om_{g,h}\Om_{gh,k}=\hp{}^g\Om_{h,k}\Om_{g,hk}.
    \end{equation}
    This constraint is derived by using associativity of the restricted symmetries: on one hand,
    \begin{align}
         (U^A_g U^A_h)U^A_k&=\Om_{g,h} U^A_{gh} U^A_k=\Om_{g,h}\Om_{gh,k} U^A_{ghk}.
    \end{align}
    On the other hand,
    \begin{align}
        U^A_g(U^A_hU^A_k)&=U^A_g\Om_{h,k}U^A_{hk}=\hp{}^g\Om_{h,k}\Om_{g,hk}U^A_{ghk}.
    \end{align}
    Comparing the final expressions, we obtain (\ref{eq:NA2cocycle}). The function $\nu:G\times G\to\mathbb{Q}_+$ is defined as follows:
    \begin{equation}
        \nu(g,h)\equiv\text{Ind}(\Om_{g,h})
    \end{equation}
    where $\text{Ind}(W)$ denotes the GNVW index of a 1D QCA $W$.\footnote{Our convention is that a GNVW index $\nu>1$ corresponds to net operator flow in the counterclockwise direction, whereas $\nu<1$ corresponds to net operator flow in the clockwise direction.} It is straightforward to verify that $\nu$ satisfies the 2-cocycle condition (\ref{eq:2cocycle}), by evaluating the index of both sides of (\ref{eq:NA2cocycle}).

    \item We now introduce an ancillary 1D qudit chain living along $\partial A$. The purpose of adding ancillas is to define modified $\Om_{g,h}$ operators, denoted by $\Om'_{g,h}$, which have trivial GNVW index but still satisfy the non-abelian 2-cocycle condition:
    \begin{equation}
        \label{eq:NA2cocycletilde}
        \Om'_{g,h}\Om'_{gh,k}=\hp{}^g\Om'_{h,k}\Om'_{g,hk}.
    \end{equation}
    To do so, the 1D chain must be composed of qudits of various dimensions. In particular, let us write $\nu(g,h)=m(g,h)/n(g,h)$ where $m(g,h)$ and $n(g,h)$ are relatively prime positive integers. Then let $\{p_1,\ldots,p_k\}$ be the set of all prime numbers appearing in the prime decomposition of $m(g,h)$ or $n(g,h)$ for any pair $g,h\in G$. Finally, define the Hilbert space 
    \begin{equation}
        \mathcal{H}_{\partial A}=\bigotimes_{r=1}^k\mathcal{H}_r\quad\text{where}\quad\mathcal{H}_r=\bigotimes_{i\in\mathbb{Z}_L}\mathbb{C}^{p_r}.
    \end{equation}
    Here, $\mathbb{Z}_L$ denotes a chain of $L$ sites living along $\partial A$. Upon adding these ancillas, the total Hilbert space becomes $\mathcal{H}'=\mathcal{H}\otimes\mathcal{H}_{\partial A}$. We now define a canonical operator $T_\nu$ on $\mathcal{H}_{\partial A}$ such that $\text{Ind}(T_\nu)=\nu$ for any GNVW index $\nu$ that can be expressed in terms of the prime factors $\{p_1,\ldots,p_k\}$. To do so, we first express $\nu=m/n$ in terms of the prime decomposition of the relatively prime positive integers $m$ and $n$:
    \begin{equation}
        m=p_{1}^{i_1}\cdots p_{k}^{i_k}\quad\text{and} \quad n=p_{1}^{j_1}\cdots p_{k}^{j_k}
    \end{equation}
    where $i_1,\ldots,i_k,j_1,\ldots,j_k\in\mathbb{Z}_{\geq 0}$. Then define
    \begin{equation}
        T_\nu=t_{p_1}^{i_1}\cdots t_{p_k}^{i_k}\left(t_{p_1}^{j_1}\cdots t_{p_k}^{j_k}\right)^{-1},
    \end{equation}
    where $t_r$ denotes a uniform translation of the qudits in $\mathcal{H}_r$ in the counterclockwise direction. Crucially, these operators satisfy the condition
    \begin{equation}
        T_\nu T_{\nu'}=T_{\nu\cdot\nu'}.\label{eq:Tcondition}
    \end{equation}
    For each pair $g,h\in G$, we define the following operator on $\mathcal{H}_{\partial A}$:
    \begin{equation}
        T_{g,h}\equiv T_{\nu(g,h)}^{-1}.
    \end{equation}
    By definition $\text{Ind}(T_{g,h})=\nu(g,h)^{-1}$. By virtue of (\ref{eq:Tcondition}) and the 2-cocycle condition (\ref{eq:2cocycle}) on $\nu:G\times G\to\mathbb{Q}_+$, it follows that
    \begin{equation}
        T_{g,h}T_{gh,k}=T_{h,k}T_{g,hk}.\label{eq:2cocycleT}
    \end{equation}
    Finally, we define the operator $\Om'_{g,h}$ on $\mathcal{H}'$ for each pair $g,h\in G$:
    \begin{equation}
        \Om'_{g,h}\equiv\Om_{g,h}\otimes T_{g,h}.
    \end{equation}
    Clearly, $\text{Ind}(\Om'_{g,h})=1$, hence $\Om'_{g,h}$ is a 1D FDQC along $\partial A$. Moreover, the non-abelian 2-cocycle condition (\ref{eq:NA2cocycletilde}) follows from (\ref{eq:NA2cocycle}) and (\ref{eq:2cocycleT}).

    \item Next, we choose a restriction $\Om^I_{g,h}$ of each $\Om'_{g,h}$ to the vicinity of the interval $I=[a,b]\subset\partial M$.

    \item For each triple $g,h,k\in G$, the $\Om^I$ operators obey the non-abelian 2-cocycle condition \textit{up to} composition by a unitary operator $\Ga_{g,h,k}$ supported near $\partial I=\{a,b\}$:
    \begin{equation}\label{eq:Gamma_def}
        \Gamma_{g,h,k}\equiv\Om^I_{g,h}\Om^I_{gh,k}\left(\hp{}^g\Om^I_{h,k}\Om^I_{g,hk}\right)^{-1}.
    \end{equation}
    For every 4-tuple $g,h,k,k\in G$, the $\Ga$ operators obey a constraint we refer to as the ``non-abelian 3-cocycle condition":
    \begin{equation}
        \Ga_{g,h,k}\cdot\hp{}^{^g(h,k)}\Ga_{g,hk,l}\cdot\hp{}^g\Ga_{h,k,l}=\hp{}
        ^{(g,h)}\Ga_{gh,k,l}\cdot\hp{}\Delta_{g,h,(k,l)}\cdot\hp{}^{^{g\cdot h}(k,l)}\Ga_{g,h,kl},\label{eq:NA3cocycle}
    \end{equation}
    where we have introduced the operator\footnote{We give $\Delta_{g,h,(k,l)}$ its subscript because of the operators that appear in its definition.}
    \begin{equation}
        \Delta_{g,h,(k,l)}\equiv\hp{}^{(g,h)\cdot gh}\Om^I_{k,l}\left(\hp{}^{g\cdot h}\Om^I_{k,l}\right)^{-1}.
        \label{eq:Delta}
    \end{equation}
    To derive this constraint, we evaluate $\Om^I_{g,h}\Om^I_{gh,k}\Om^I_{ghk,l}$ in two different ways. On one hand,
    \begin{align}
        \Om^I_{g,h}\Om^I_{gh,k}\Om^I_{ghk,l}&=\Ga_{g,h,k}\cdot\hp{}^g\Om^I_{h,k}\Om^I_{g,hk}\Om^I_{ghk,l}\\
        &=\Ga_{g,h,k}\cdot\hp{}^{^g(h,k)}\Gamma_{g,hk,l}\cdot\hp{}^g(\Om^I_{h,k}\Om^I_{hk,l})\Om^I_{g,hkl}\\
        &=\Ga_{g,h,k}\cdot\hp{}^{^g(h,k)}\Ga_{g,hk,l}\cdot\hp{}^g\Ga_{h,k,l}\cdot\hp{}^{g\cdot h}\Om^I_{k,l}\cdot\hp{}^g\Om^I_{h,kl}\Om^I_{g,hkl}.\label{eq:3cocycle1}
    \end{align}
    On the other hand,
    \begin{align}
        \Om^I_{g,h}\Om^I_{gh,k}\Om^I_{ghk,l}&=\hp{}^{(g,h)}\Ga_{gh,k,l}\cdot\hp{}^{(g,h)\cdot gh}\Om^I_{k,l}\Om^I_{g,h}\Om^I_{gh,kl}\\
        &=\hp{}^{(g,h)}\Ga_{gh,k,l}\cdot\Delta_{g,h,(k,l)}\cdot\hp{}^{g\cdot h}\Om^I_{k,l}\Om^I_{g,h}\Om^I_{gh,kl}\\
        &=\hp{}^{(g,h)}\Ga_{gh,k,l}\cdot\Delta_{g,h,(k,l)}\cdot\hp{}^{^{g\cdot h}(k,l)}\Ga_{g,h,kl}\cdot\hp{}^{g\cdot h}\Om^I_{k,l}\cdot\hp{}^g\Om^I_{h,kl}\Om^I_{g,hkl}.
        \label{eq:3cocycle2}
    \end{align}
    Comparing (\ref{eq:3cocycle1}) and (\ref{eq:3cocycle2}), we obtain (\ref{eq:NA3cocycle}).

    \item Next, we choose a restriction $\Ga^a_{g,h,k}$ of $\Ga_{g,h,k}$ to the vicinity of point $a$. Clearly, this choice is unique up to multiplication by a $U(1)$ phase.

    \item We then define an operator $\Delta^a_{g,h,(k,l)}$, which is a restriction of $\Delta_{g,h,(k,l)}$ to the vicinity of point $a$. Unlike $\Ga^a_{g,h,k}$, the overall phase of $\Delta^a_{g,h,(k,l)}$ is not arbitrarily chosen; we may define it canonically. To do so, we first choose an arbitrary point $c$ in the interior of $I$, and define the intervals $L=[a,c]$ and $R=[c,b]$. Then, we arbitrarily decompose the operator $\Om^I_{k,l}$ as a product of two 1D FDQCs, one supported in the vicinity of $L$ and one supported in the vicinity of $R$. That is, $\Om^I_{k,l}=\Om^L_{k,l}\Om^R_{k,l}$.
    We then define
    \begin{align}\label{eq:Delta_a_def}
        \Delta^a_{g,h,(k,l)}&\equiv\hp{}^{(g,h)\cdot gh}\Om^L_{k,l}\hp{}\left(\hp{}^{g\cdot h}\Om^L_{k,l}\right)^{-1}.
    \end{align}
    Clearly this definition is independent of the choice of intervals $L$ and $R$, and the choice of operators $\Om^L_{g,h}$ and $\Om^R_{g,h}$.\footnote{We note that $\Delta^a_{g,h,(k,l)}$ can also be expressed using the ``commutator pairing" introduced in \cite{Kapustin}. In particular, let $d$ be a point to the left of $a$, and define the interval $J=[d,a]$. The commutator pairing between a 1D FDQC $A$ supported near $J$ and a 1D FDQC $B$ supported near $L=[a,c]$ is defined as the unitary $\eta(A,B)=ABA^{-1}B^{-1}$, which is supported near point $a$. Choose a restriction $\Om^K_{g,h}$ of $\Om_{g,h}$ to $K=J\cup L=[d,c]$ which coincides with $\Om^L_{g,h}$ near point $c$, and let $\Om^J_{g,h}=\Om^K_{g,h}(\Om^L_{g,h})^{-1}$. Then, via a simple calculation, we find that $\Delta^a_{g,h,(k,l)}=\eta(A,B)^{-1}$ where $A=\Om^J_{g,h}$ and $B=\hp{}^{(g,h)\cdot gh}\Om^L_{k,l}$.}

    \item Since the $\Gamma$ and $\Delta$ operators satisfy the non-abelian 3-cocycle condition (\ref{eq:3cocycle}), the corresponding $\Ga^a$ and $\Delta^a$ operators must obey an analogous constraint, up to a $U(1)$ phase. The function $\omega:G\times G\times G\times G\to U(1)$ is defined in terms of this phase. Specifically,
    \begin{equation}
        \omega(g,h,k,l)\equiv\Ga^a_{g,h,k}\cdot\hp{}^{^g(h,k)}\Ga^a_{g,hk,l}\cdot\hp{}^g\Ga^a_{h,k,l}\left(\hp{}^{(g,h)}\Ga^a_{gh,k,l}\cdot\hp{}\Delta^a_{g,h,(k,l)}\cdot\hp{}^{^{g\cdot h}(k,l)}\Ga^a_{g,h,kl}\right)^{-1}.
        \label{eq:omega}
    \end{equation}
    In Appendix~\ref{app:cocycle}, we demonstrate that $\omega$ is a 4-cocycle, \textit{i.e.} it satisfies the 4-cocycle condition (\ref{eq:4cocycle}).
\end{enumerate}

We have thus arrived at a formula for the anomaly 4-cocycle $\omega$, in terms of the operators $U_g^A$, $\Om^I_{g,h}$, $\Ga^a_{g,h,k}$, and $\Delta^a_{g,h,(k,l)}$. While this is a concrete expression, it may appear somewhat arcane. To illuminate its structure, in Section~\ref{diagrams} we provide a graphical representation of the formula and its derivation.

\subsection{Ambiguity of the output cocycles}
\label{sec:ambiguity}

Thus far, we have explained how to define a $\mathbb{Q}_+$-valued 2-cocycle $\nu$, and a $U(1)$-valued 4-cocycle $\omega$, via a series of concrete steps. However, there are arbitrary choices made at various points in this procedure, which all affect the cocycles that are obtained. We now explain why different choices always lead to cocycles belonging to the same cohomology classes in $H^2(G,\mathbb{Q}_+)$ and $H^4(G,U(1))$, respectively. We will also see why our procedure can produce \textit{any} pair of representative cocycles in their respective equivalence classes.
Thus, the invariants we compute are precisely elements of $H^2(G,\mathbb{Q}_+)$ and $H^4(G,U(1))$.

We will work backward, starting from the end of the procedure and working toward the start. The last choice made is the choice of restriction of $\Ga_{g,h,k}$ to the vicinity of point $a$, which has an inherent phase ambiguity. That is, instead of $\Ga^a_{g,h,k}$, we are free to instead choose $\rho(g,h,k)\Ga^a_{g,h,k}$ for an arbitrary $U(1)$ phase $\rho(g,h,k)$. Given this choice, instead of $\omega(g,h,k,l)$, we would obtain the 4-cocycle
\begin{equation}
    \tilde\omega(g,h,k,l)=\omega(g,h,kl)\frac{\rho(g,h,k)\rho(g,hk,l)\rho(h,k,l)}{\rho(gh,k,l)\rho(g,h,kl)}.
\end{equation}
We find that $\tilde\omega$ differs from $\omega$ by the 4-coboundary $d\rho$. Thus, different choices of $\Ga^a_{g,h,k}$ give rise to equivalent 4-cocycles.\footnote{We note that there is crucially no such ambiguity in the restriction $\Delta^a_{g,h,(k,l)}$ of $\Delta_{g,h,(k,l)}$, which is canonically defined, including its overall phase.}
Furthermore, since $\rho$ may be any 3-cochain, $\tilde{\omega}$ may be any 4-cocycle which is equivalent to $\omega$.

Continuing to work backward, there is a choice of restriction $\Om^I_{g,h}$ of $\Om'_{g,h}$. We show in Appendix~\ref{appendix:ambiguity1} that the freedom in this choice does not contribute any additional ambiguity to $\omega$ beyond the 4-coboundary ambiguity arising from the restriction of $\Ga_{g,h,k}$. Thus, different choices of $\Om^I_{g,h}$ give rise to equivalent 4-cocycles.

Next, consider the choice of symmetry restriction $U^A_g$. Instead of $U^A_g$, we could instead choose $\tilde U^A_g=\Sig_gU^A_g$ where $\Sig_g$ is an arbitrary 1D QCA acting along $\partial A$. Given this choice, instead of $\nu(g,h)$ we would obtain the 2-cocycle
\begin{equation}
    \tilde\nu(g,h)=\nu(g,h)\frac{\mu(g)\mu(h)}{\mu(gh)}
\end{equation}
where $\mu(g)=\text{Ind}(\Sig_g)$. We find that $\tilde\nu$ differs from $\nu$ by the 2-coboundary $d\mu$. Thus, different choices of symmetry restriction give rise to equivalent 2-cocycles. Similarly to before, since $\mu$ may be any 1-cochain, $\tilde\nu$ may be any $2$-cocycle which is equivalent to $\nu$. We show in Appendix~\ref{appendix:ambiguity2} that different choices of symmetry restriction also give rise to equivalent 4-cocycles.

Finally, we consider the choice of disk $A$ and interval $I\in\partial A$. Clearly, small deformations of these regions can be absorbed into the choice of restrictions $U_g^A$ and $\Om^I_{g,h}$. The cohomology classes of the output cocycles are thus invariant under such small deformations, and it follows that they are also invariant under large deformations. We conclude that our procedure yields well-defined cohomology classes $[\nu]\in H^2(G,\mathbb{Q}_+)$ and $[\omega]\in H^4(G,U(1))$. Clearly, these indices are additive under stacking of $G$-symmetric systems.

\section{Obstruction to a trivially gapped Hamiltonian}
\label{sec:obstruction}

In the previous section, we have defined a pair of indices, $[\nu]\in H^2(G,\mathbb{Q}_+)$ and $[\omega]\in H^4(G,U(1))$, for any $G$-symmetry acting by FDQCs on a 2D lattice system. Here, we justify the interpretation of $[\omega]$ as a label for the anomaly of the symmetry. That is, we show---under certain assumptions---that a $G$-symmetry with nontrivial $[\omega]$ does not admit a $G$-symmetric invertible state. 
Thus, there cannot exist a symmetric, gapped, local Hamiltonian with a unique invertible ground state.

We will work under the assumption that the defining procedure for the index $[\omega]$ can be appropriately modified to the setting of QCAs and FDQCs ``with tails". More precisely, a QCA with tails refers to a unitary that preserves locality of operators up to corrections that decay exponentially with distance. Similarly, a FDQC with tails refers to a circuit whose gates are supported on a finite number of sites, up to exponentially decaying corrections.\footnote{We note that a FDQC with tails is equivalent to a locally generated unitary, \textit{i.e.} finite time evolution by a time-dependent local Hamiltonian.} In 1D, the GNVW index is well-defined in the setting of QCAs with tails, and the $\mathbb{Q}_+$ classification of 1D QCAs modulo FDQCs is unaffected by allowing for tails \cite{Ranard}. Thus, we anticipate that  the procedure can be modified to allow for non-strict locality of all of the operators used to define $[\nu]$ and $[\omega]$.\footnote{This would require a suitable notion of ``restriction" of a QCA with tails.} We emphasize that there are subtleties involving finite size effects arising from the lack of strict locality; we do not attempt a rigorous treatment here. Henceforth, the terms QCA and FDQC will refer to operators with tails.

Bearing these subtleties in mind, we may now outline the argument. Recall that a state $\ket{\psi}$ is \textit{invertible} if there exists a state $\ket{\phi}$ living in an ancillary Hilbert space, and a FDQC $V$ such that $V(\ket{\psi}\otimes\ket{\phi})$ is a product state \cite{Kitaev1,Kitaev2}. Further recall that a state $\ket{\psi}$ is \textit{short-range entangled} (SRE) if there exists a FDQC $W$ such that $W(\ket{\psi}\otimes\ket{\psi_0})$ is a product state, where $\ket{\psi_0}$ is itself an ancillary product state \cite{LRE}.
Note that in general, short-range entanglement is a strictly stronger condition than invertibility.\footnote{We note that Kitaev uses an alternative definition of SRE that includes invertible states in the class of SRE states. Here, we use the stated definition.} However, in 1D bosonic systems it has been shown that all invertible states are also SRE \cite{KapustinYangSopenko}.\footnote{In 1D fermionic systems there are invertible states that are not SRE, such as the Kitaev chain \cite{KitaevWire}.}

In the following proof, we will ignore the need to tensor with ancillas to trivialize SRE states. In other words, for any SRE state $\ket{\psi}$, we assume there exists a FDQC $W$ such that $W\ket{\psi}$ is a product state. This assumption is not strictly necessary, and is merely invoked to streamline the proof.

We now demonstrate the following theorem:
\begin{thm}\label{thm}
    Suppose $\{U_g\}$ is a $G$-symmetry acting by FDQCs on a 2D lattice system, and there exists a $G$-symmetric invertible state $\ket{\psi}$. Then, $\{U_g\}$ has trivial anomaly index $[\omega]\in H^4(G,U(1))$.
\end{thm}
\begin{proof}
    The proof follows from the three lemmas below. First, consider the case where $\ket{\psi}$ is not only invertible, but is also SRE.
    
    To begin, consider a restriction $U^A_g$ of $U_g$ to the disk $A$. By Lemma~\ref{lemma1}, for each $g\in G$ there exists a FDQC $\Sigma_g$ supported near $\partial A$ such that $\Sigma_gU^A_g\ket{\psi}=\ket{\psi}$. Let us redefine $U^A_g\to \Sigma_gU^A_g$, such that after the redefinition, $U^A_g\ket{\psi}=\ket{\psi}$. Since $\Om_{g,h}=U^A_gU^A_h(U^A_{gh})^{-1}$, it follows that $\Om_{g,h}\ket{\psi}=\ket{\psi}$.

    Next, we define the expanded Hilbert space $\mathcal{H}'=\mathcal{H}\otimes\mathcal{H}_{\partial A}$ and the operator $\Om'_{g,h}=\Om_{g,h}\otimes T_{g,h}$ as in step 4 of our anomaly computation procedure. Moreover, we define a state $\ket{\psi'}=\ket{\psi}\otimes\ket{\psi_{\partial A}}$ where $\ket{\psi_{\partial A}}$ is the 1D translation invariant product state from step 4. Since $\ket{\psi_{\partial A}}$ is translation invariant, it follows that $T_{g,h}\ket{\psi_{\partial A}}=\ket{\psi_{\partial A}}$, and hence $\Om'_{g,h}\ket{\psi'}=\ket{\psi'}$. We now consider a restriction $\Om^I_{g,h}$ of $\Om'_{g,h}$ to the interval $I=[a,b]\in\partial A$. By Lemma~\ref{lemma2}, for each pair $g,h\in G$ there exists a FDQC $\La_{g,h}$ supported near points $a$ and $b$ such that $\La_{g,h}\Om^I_{g,h}\ket{\psi'}=\ket{\psi'}$. Let us redefine $\Om^I_{g,h}\to\La_{g,h}\Om^I_{g,h}$, such that $\Om^I_{g,h}\ket{\psi'}=\ket{\psi'}$. By definitions (\ref{eq:Gamma_def}) and (\ref{eq:Delta}), it follows that $\Ga_{g,h,k}\ket{\psi'}=\Delta_{g,h,(k,l)}\ket{\psi'}=\ket{\psi'}$.
    
    Since $\Om^I_{g,h}\ket{\psi'}=\ket{\psi'}$, it follows by Lemma~\ref{lemma2} that there exists a decomposition $\Om^I_{g,h}=\Om^L_{g,h}\Om^R_{g,h}$, where $\Om^L_{g,h}$ ($\Om^R_{g,h}$) is supported near the interval $L$ ($R$), such that $\Om^L_{g,h}\ket{\psi'}=\ket{\psi'}$. By definition (\ref{eq:Delta_a_def}), it follows that $\Delta^a_{g,h,(k,l)}\ket{\psi'}=\ket{\psi'}$.
    
    By Lemma~\ref{lemma3}, there exists a restriction $\Ga^a_{g,h,k}$ of $\Ga_{g,h,k}$ to the vicinity of $a$ such that $\Ga^a_{g,h,k}\ket{\psi'}=\ket{\psi'}$. Altogether, by definition (\ref{eq:omega}), it follows that $\omega(g,h,k,l)\ket{\psi'}=\ket{\psi'}$. Hence $\omega(g,h,k,l)=1$, and the anomaly index $[\omega]$ is trivial.

    Thus far, we have shown that the theorem holds whenever $\ket{\psi}$ is SRE. We now claim that the general case, in which $\ket{\psi}$ is merely invertible, immediately follows. To see this, consider the SRE state $\ket{\psi}_\text{SRE}=\ket{\psi}\otimes\ket{\phi}$ which is guaranteed to exist by invertibility of $\ket{\psi}$. Clearly $\ket{\psi}_\text{SRE}$ is symmetric under the $G$-symmetry $\{U_g\otimes\mathbbm{1}\}$. Therefore, $\{U_g\otimes\mathbbm{1}\}$ has trivial anomaly index, hence $\{U_g\}$ must also have trivial anomaly index.
\end{proof}

\begin{figure}
    \centering
    \begin{tikzpicture}
        \filldraw[thick,dashed,black,fill=red!15](0,0) circle (3.6);
        \draw[thick,black](-.4177,-3.1726) arc(262.5:-82.5:3.2);
        \filldraw[thick,dashed,black,fill=blue!15](0,0) circle (2.8);
        \draw[thick,black](-.4168,-2.3635) arc(260:-80:2.4);

        \filldraw[thick,dashed,black,fill=white](0,0) circle (2);
        
        \node at (0,3){$A$};
        \node at (0,2.2){$B$};
        \node at (0,-1.7){$\text{Int}$};
        \node at (0,-2.4){$\boldsymbol{\partial B}$};
        \node at (0,-3.2){$\boldsymbol{\partial A}$};
        \node at (0,-3.9){$\text{Ext}$};
    \end{tikzpicture}
    \caption{Geometry used in the proof of Lemma~\ref{lemma1}. Here, $A$ is the disk whose boundary is the outer solid line, and $B$ is the disk whose boundary is the inner solid line. The region $\text{Ext}$ lies outside the outermost dashed circle, and the region $\text{Int}$ lies inside the innermost dashed circle. The annulus $\boldsymbol{\partial A}$ lies between the middle and outer dashed lines (shaded red region), whereas the annulus $\boldsymbol{\partial B}$ lies between the inner and middle dashed lines (shaded blue region). Note that we use bold font to distinguish between the boundary curves $\partial A$ and $\partial B$, and the surrounding annuli $\boldsymbol{\partial A}$ and $\boldsymbol{\partial B}$.}
    \label{fig:annulus}
\end{figure}
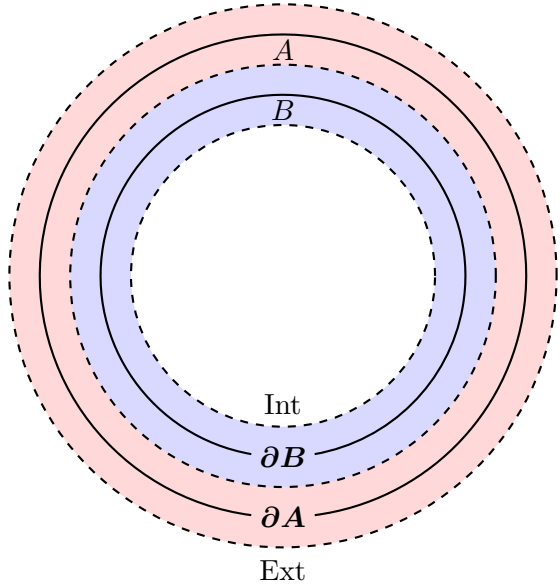

\begin{lem}
    \label{lemma1}
    Given a 2D SRE state $\ket{\psi}$, a FDQC $U$ such that $U\ket{\psi}=\ket{\psi}$, and a restriction $U^A$ of $U$ to a disk $A$, there exists a FDQC $\Sigma$ supported near $\partial A$ such that $\Sigma U^A\ket{\psi}=\ket{\psi}$.\footnote{Again, we will ignore the need to tensor with ancillas to trivialize SRE states.}${}^{,}$\footnote{We thank Michael Levin for insight on the proof of this lemma.}
\end{lem}
\begin{proof}
    This proof uses the fact that every 1D invertible state is SRE, as shown in \cite{KapustinYangSopenko}. First, consider the case where $\ket{\psi}$ is a product state $\ket{\psi}=\bigotimes_i\ket{\psi}_i$.

    Consider the geometry shown in Fig.~\ref{fig:annulus}. The basic idea of the proof is to decompose the total Hilbert space as a tensor product over the four regions $\text{Ext}$, $\boldsymbol{\partial A}$, $\boldsymbol{\partial B}$, and $\text{Int}$. That is,
    \begin{equation}
        \mathcal{H}=\mathcal{H}_\text{Ext}\otimes\mathcal{H}_{\boldsymbol{\partial A}}\otimes\mathcal{H}_{\boldsymbol{\partial B}}\otimes\mathcal{H}_\text{Int}.
    \end{equation}
    For any region $R$, define the state $\ket{\psi}_R=\bigotimes_{i\in R}\ket{\psi}_i$. Now consider a restriction $U^{A\setminus B}$ of $U$ to the annulus $A\setminus B$, which is identical to $U^A$ along the boundary of $A$.\footnote{More precisely, we require that $U^{A\setminus B}(U^A)^{-1}$ is supported near $B$.} The state $\ket{\psi'}=U^{A\setminus B}\ket{\psi}$ looks like $\ket{\psi}$ deep within disk $B$, deep within the annulus $A\setminus B$, and far outside disk $A$. Therefore, we may decompose $\ket{\psi'}$ as a tensor product
    \begin{equation}
        \ket{\psi'}=\ket{\psi}_\text{Ext}\otimes\ket{\phi}_{\boldsymbol{\partial A}}\otimes\ket{\phi}_{\boldsymbol{\partial B}}\otimes\ket{\psi}_\text{Int}
    \end{equation}
    where $\ket{\psi}_\text{Ext}=\bigotimes_{i\in \text{Ext}}\ket{\psi}_i$ and $\ket{\psi}_\text{Int}=\bigotimes_{i\in\text{Int}}\ket{\psi}_i$.

    The next step is to view regions $\boldsymbol{\partial A}$ and $\boldsymbol{\partial B}$ as standalone 1D systems, and $\ket{\phi}_{\boldsymbol{\partial A}}$ and $\ket{\phi}_{\boldsymbol{\partial B}}$ as 1D states living therein. By definition,
    \begin{equation}
        (U^{A\setminus B})^{-1}(\ket{\phi}_{\boldsymbol{\partial A}}\otimes\ket{\phi}_{\boldsymbol{\partial B}})=\ket{\psi}_{\boldsymbol{\partial A}\cup\boldsymbol{\partial B}}.
    \end{equation}
    Therefore, $\ket{\phi}_{\boldsymbol{\partial A}}$ is a 1D invertible state, hence it is SRE. Consequently, there exists a FDQC unitary $\Sigma$, supported near $\boldsymbol{\partial A}$, such that $\Sigma\ket{\phi}_{\boldsymbol{\partial A}}=\ket{\psi}_{\boldsymbol{\partial A}}$. Since
    \begin{equation}
        U^A\ket{\psi}=\ket{\psi}_\text{Ext}\otimes\ket{\phi}_{\boldsymbol{\partial A}}\otimes\ket{\psi}_{\boldsymbol{\partial B}}\otimes\ket{\psi}_\text{Int},
    \end{equation}
    it follows that $\Sigma U^A\ket{\psi}=\ket{\psi}$, as desired.
    
    Thus far, we have shown that this lemma holds if $\ket{\psi}$ is a product state. If, more generally, $\ket{\psi}$ is an arbitrary SRE state, then there exists a FDQC $V$ such that $V\ket{\psi}=\ket{\psi_0}$ where $\ket{\psi_0}$ is a product state. Thus $U'\ket{\psi_0}=\ket{\psi_0}$ where $U'=VUV^\dagger$. The operator $VU^AV^\dagger$ is a restriction of $U'$ to region $A$, hence there exists an FDQC $\Sigma'$ such that $\Sigma'VU^AV^\dagger\ket{\psi_0}=\ket{\psi_0}$. Thus, $\Sigma U^A\ket{\psi}=\ket{\psi}$ where $\Sigma=V^\dagger\Sigma'V$.
\end{proof}
\begin{lem}
    \label{lemma2}
    Given a 2D SRE state $\ket{\psi}$, a 1D FDQC $\Om$ supported near a closed curve $\gamma$ such that $\Om\ket{\psi}=\ket{\psi}$, and a restriction $\Omega^I$ of $\Om$ to an interval $I=[a,b]\in\gamma$, there exists a 1D FDQC $\Lambda$, supported in the vicinity of points $a$ and $b$, such that $\Lambda \Om^I\ket{\psi}=\ket{\psi}$.
\end{lem}
\begin{proof}
    First, consider the case where $\ket{\psi}$ is a product state $\ket{\psi}=\bigotimes_i\ket{\psi}_i$. Then $\ket{\psi'}=\Om^I\ket{\psi}$ can be decomposed as a tensor product over three regions: \begin{equation}
        \ket{\psi'}=\ket{\phi}_{\boldsymbol{a}}\otimes\ket{\phi}_{\boldsymbol{b}}\otimes\ket{\psi}_C
    \end{equation}
    where $\boldsymbol{a}$ ($\boldsymbol{b}$) is a small disk around point $a$ ($b$), and $C$ is the complement of these two disks. Here $\ket{\psi}_C=\bigotimes_{i\in C}\ket{\psi}_i$. Let $\La^a$ be a unitary operator supported in $\boldsymbol{a}$ such that $\La^a\ket{\phi}_{\boldsymbol{a}}=\ket{\psi}_{\boldsymbol{a}}=\bigotimes_{i\in\boldsymbol{a}}\ket{\psi}_i$, and similarly for $\La^b$. Setting $\La=\La^a\La^b$, we have that $\La\Om^I\ket{\psi}=\ket{\psi}$ as desired.

    We have shown that the lemma holds if $\ket{\psi}$ is a product state. If $\ket{\psi}$ is an arbitrary SRE state, it holds via the same reasoning as in the proof of the previous lemma. Physically, this lemma states that a gapped parent Hamiltonian for $\ket{\psi}$ cannot support nontrivial anyonic excitations.
\end{proof}
\begin{lem}
    \label{lemma3}
    Given a 2D SRE state $\ket{\psi}$, a FDQC $\Ga$ supported near a pair of well-separated points $\{a,b\}$ such that $\Ga\ket{\psi}=\ket{\psi}$, and a restriction $\Ga^a$ of $\Ga$ to the vicinity of $a$, there exists a $U(1)$ phase $\rho$ such that $\rho\Ga^a\ket{\psi}=\ket{\psi}$.
\end{lem}
\begin{proof}
    Again, consider the case where $\ket{\psi}$ is a product state. Divide the system into two parts, region 1 containing point $a$ and region 2 containing point $b$. Clearly $\ket{\psi'}=\Gamma\ket{\psi}$ can be decomposed as a tensor product state $\ket{\psi'}=\ket{\psi_1}\otimes\ket{\psi_2}$. Since there exists some $\rho\in U(1)$ such that $\rho\Ga^a\ket{\psi_1}=\ket{\psi_1}$, it follows that $\rho\Ga^a\ket{\psi}=\ket{\psi}$.

    We have shown that the lemma holds if $\ket{\psi}$ is a product state. If $\ket{\psi}$ is an arbitrary SRE state, it holds via the same reasoning as in the proof of the previous lemmas.
\end{proof}

\section {Example: Anomalous $\mathbb{Z}_2\times\mathbb{Z}_2\times\mathbb{Z}_2\times\mathbb{Z}_2$ symmetry}
\label{sec:example}

In this section, we illustrate our anomaly computation procedure via a simple example. Consider a system composed of a single qubit on each site $i$ of a triangular lattice, which is divided into three sublattices labeled $\Lambda_1,\Lambda_2,\Lambda_3$. Denote the Pauli operators on site $i$ by $X_i$ and $Z_i$. We will compute the anomaly of a particular $G=\mathbb{Z}_2\times\mathbb{Z}_2\times\mathbb{Z}_2\times\mathbb{Z}_2$ symmetry on this system. Let us denote each group element by a 4-tuple $\mathbf{g}=(g_1,g_2,g_3,g_4)$ with $g_j\in\{0,1\}$. The symmetry generators have the following form \cite{Yoshida}:
\begin{align}
    U_{(1,0,0,0)}&=\prod_{i\in\Lambda_1}X_i,\label{eq:Z2_1}\\
    U_{(0,1,0,0)}&=\prod_{i\in\Lambda_2}X_i,\\
    U_{(0,0,1,0)}&=\prod_{i\in\Lambda_3}X_i,\label{eq:Z2_3}\\
    U_{(0,0,0,1)}&=\prod_{\langle ijk\rangle}CCZ_{ijk}.\label{eq:Z2_4}
\end{align}
Here, $\langle ijk\rangle$ indexes the set of all elementary triangles of the lattice (each elementary triangle contains one vertex belonging to each of the $\Lambda_1,\Lambda_2,\Lambda_3$ sublattices), and the operator
$CCZ_{ijk}=(-1)^{(1-Z_i)(1-Z_j)(1-Z_k)/8}$ is the controlled-controlled-$Z$ gate acting on sites $i,j,k$. To verify that these generators mutually commute, recall the following relations:
\begin{equation}
\begin{split}
    X_iCCZ_{ijk}&=CZ_{jk}CCZ_{ijk}X_i,\\
    X_jCCZ_{ijk}&=CZ_{ki}CCZ_{ijk}X_j,\\
    X_kCCZ_{ijk}&=CZ_{ij}CCZ_{ijk}X_k,
\end{split}
\end{equation}
where $CZ_{ij}=(-1)^{(1-Z_i)(1-Z_j)/4}$ is the controlled-$Z$ gate acting on sites $i,j$.
The symmetry $U_{(0,0,0,1)}$ squares to the identity because the $CCZ$ gates individually square to the identity and commute with one another.

\begin{figure}[t]
\centering
\includegraphics[width=.35\columnwidth]{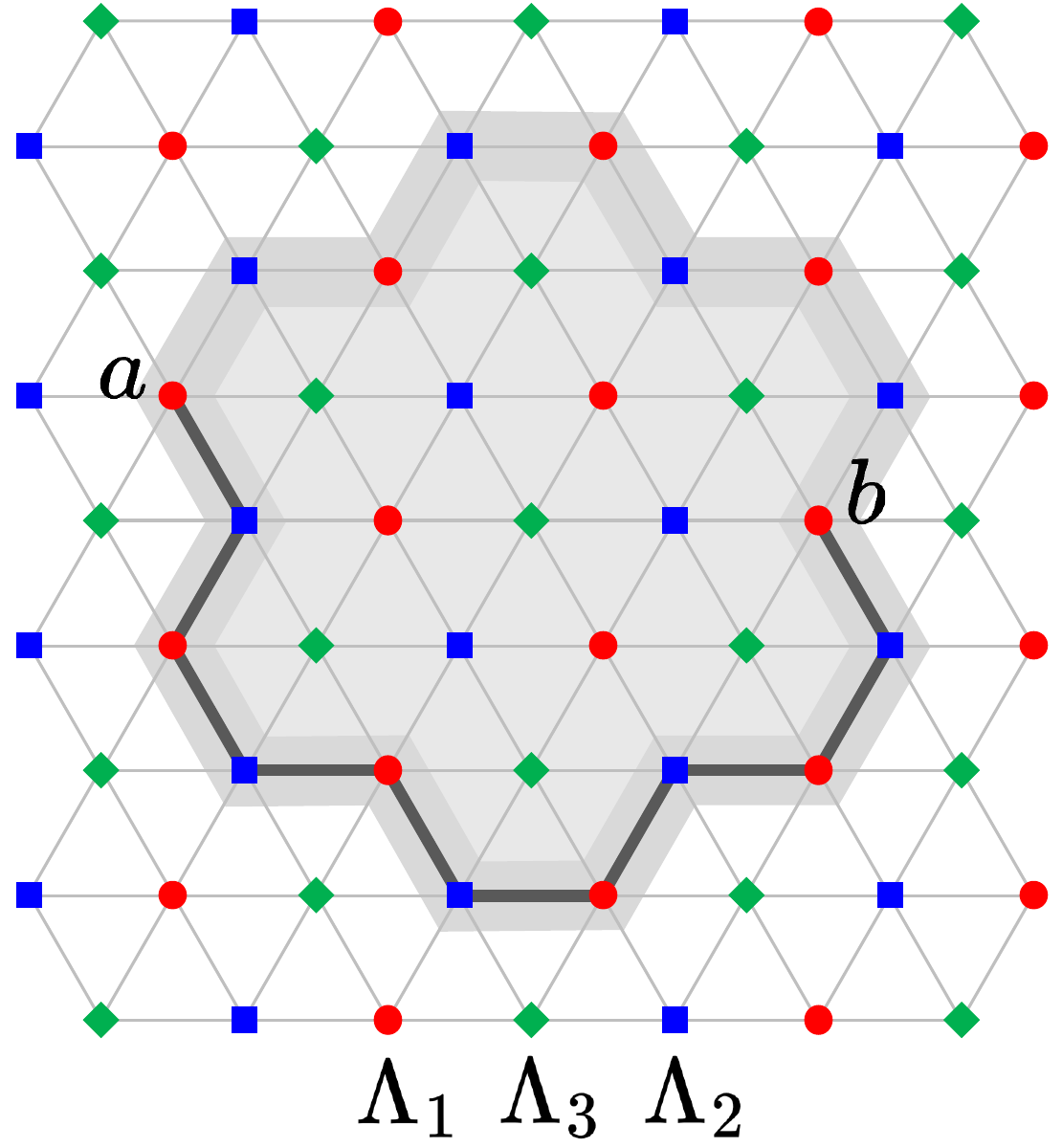}
\caption{Triangular lattice divided into three triangular sublattices $\Lambda_1,\Lambda_2,\Lambda_3$, respectively represented by red dots, blue squares, and green diamonds. All sites within the large shaded region belong to $A$, and those in the shaded outer strip belong to $\partial A$. The 11 sites along the thickened edges belong to the interval $I\subset\partial A$. Sites $a$ and $b$ are as labeled.}
\label{fig:ex}
\end{figure}

Let us define a disk $A$ whose outermost sites belong only to sublattices $\Lambda_1$ and $\Lambda_2$. We will regard a site as belonging to $\partial A$ if at least one of its nearest neighbors lies outside $A$. Furthermore, we define an interval $I\subset\partial A$ whose endpoints $a$ and $b$ belong to sublattice $\Lambda_1$. For instance, consider the geometry depicted in Fig.~\ref{fig:ex}.

To compute the anomaly index, we first choose the following set of restricted symmetry operators:
\begin{align}
    U^A_\mathbf{g}&=\left(\prod_{\langle ijk\rangle\subset A}CCZ_{ijk}\right)^{g_4}\left(\prod_{i\in\Lambda_1\cap A}X_i\right)^{g_1}\left(\prod_{i\in\Lambda_2\cap A}X_i\right)^{g_2}\left(\prod_{i\in\Lambda_3\cap A}X_i\right)^{g_3}.
\end{align}
As a result of this choice, we find that
\begin{equation}
    \Om_{\mathbf{g},\mathbf{h}}
    =\left(\prod_{\langle ij\rangle\subset\partial A}CZ_{ij}\right)^{g_3h_4}
\end{equation}
where $\langle ij\rangle$ indexes the set of all links of the lattice. $\Om_{\mathbf{g},\mathbf{h}}$ has trivial GNVW index since it is an FDQC, so there is no need to introduce the ancillary Hilbert space $\mathcal{H}_{\partial A}$. We choose the following restriction of $\Om_{\mathbf{g},\mathbf{h}}$ to the interval $I$:
\begin{equation}
    \Om^I_{\mathbf{g},\mathbf{h}}=\left(\prod_{\langle ij\rangle\subset I}CZ_{ij}\right)^{g_3h_4}.
\end{equation}
Consequently, we find that
\begin{equation}\label{eqn:gammaexample}
    \Ga_{\mathbf{g},\mathbf{h},\mathbf{k}}
    =(Z_aZ_b)^{g_2h_3k_4}.
\end{equation}
To see this, note that in this example $\Om^I_{\mathbf{g},\mathbf{h}}\Om^I_{\mathbf{gh},\mathbf{k}}=\Om^I_{\mathbf{h},\mathbf{k}}\Om^I_{\mathbf{g},\mathbf{hk}}$, so $\Ga_{\mathbf{g},\mathbf{h},\mathbf{k}}=\Omega^I_{\mathbf{h},\mathbf{k}}({}^{\mathbf{g}}\Omega^I_{\mathbf{h},\mathbf{k}})^{-1}$.
Equation (\ref{eqn:gammaexample}) then follows from the relations
\begin{equation}
\begin{split}
    X_iCZ_{ij}&=Z_jCZ_{ij}X_i,\\
    X_jCZ_{ij}&=Z_iCZ_{ij}X_j.
\end{split}
\end{equation}
Moreover, we find that $\Delta_{\mathbf{g},\mathbf{h},(\mathbf{k},\mathbf{l})}$, as well as $\Delta^a_{\mathbf{g},\mathbf{h},(\mathbf{k},\mathbf{l})}$, are the identity operator for all $\mathbf{g},\mathbf{h},\mathbf{k},\mathbf{l}\in G$. Proceeding onward, we restrict $\Ga_{\mathbf{g},\mathbf{h},\mathbf{k}}$ to
\begin{equation}
    \Ga^a_{\mathbf{g},\mathbf{h},\mathbf{k}}=(Z_a)^{g_2h_3k_4}.
\end{equation}
Finally, the formula (\ref{eq:omega}) yields the 4-cocycle $\omega:G\times G\times G\times G\to U(1)$ with values
\begin{equation}
    \omega(\mathbf{g},\mathbf{h},\mathbf{k},\mathbf{l})=(-1)^{g_1h_2k_3l_4}.
\end{equation}
This 4-cocycle corresponds to the ``type-IV" anomaly of the symmetry group $G$ \cite{Propitius,WangLevin}. The symmetry generators (\ref{eq:Z2_1}-\ref{eq:Z2_4}) have a simple interpretation that is consistent with this result: the first three generators (\ref{eq:Z2_1}-\ref{eq:Z2_3}) constitute an on-site $\mathbb{Z}_2\times\mathbb{Z}_2\times\mathbb{Z}_2$ symmetry, and the fourth generator (\ref{eq:Z2_4}) is a symmetric entangler for the 2D ``type-III" SPT protected by this on-site symmetry \cite{DecoratedDomainWall}.\footnote{A symmetric entangler for an SPT is a locality preserving unitary which commutes with the symmetry operators and maps a product state into the SPT ground state.}

\section{Diagrammatic representation of the anomaly formula}
\label{diagrams}

In Sec.~\ref{sec:index}, we gave explicit algebraic definitions of the operators $\Omega_{g,h}$, $\Gamma_{g,h,k}$, and $\Delta_{g,h,(k,l)}$, culminating in a formula for the anomaly 4-cocycle $\omega(g,h,k,l)$.
Although not arbitrary, these expressions may seem somewhat arcane. 
In this section, we give graphical representations of the defining expressions for each of these objects, as well as the non-abelian 2-cocycle and 3-cocycle conditions that they satisfy, in an effort to demystify our procedure.
For simplicity, we take the $H^2(G,\mathbb{Q}_+)$ index of the symmetry to be trivial.

Our diagrammatic formalism is reminiscent of structures from higher category theory: Just as pasting diagrams encode relationships between ordered compositions of morphisms, the diagrams we draw encode relationships between ordered compositions of operators. Although we do not use the categorical language in our discussion, we expect that this analogy can be made precise and generalized to higher dimensions.

We begin by considering the original symmetry operators $\{U_g\}$, which are subject to the group multiplication law $U_gU_h=U_{gh}$. This equation may be viewed as a ``non-abelian 1-cocycle condition", and is represented graphically as a 1-cube:
\begin{center}
    \begin{tikzpicture}
        \node at (0,0){$U_gU_h$};
        \node at (0,2){$U_{gh}$};
        \node[rotate=90] at (0,1){$=$};
    \end{tikzpicture}
\end{center}

Restricting these symmetry operators to the region $A$, this constraint becomes the defining formula for $\Omega_{g,h}$:
\begin{equation}
    \Omega_{g,h}U^A_{gh}=U_g^AU_h^A.
\end{equation}
After the symmetry restriction, the 1-cube now represents this operator:
\begin{center}
    \begin{tikzpicture}
        \node at (0,0){$U_g^AU_h^A$};
        \node at (0,2){$U_{gh}^A$};
        \draw[thick,black,->](0,1.6)--(0,.4);
        \node at (-.5,1){$\Omega_{g,h}$};
    \end{tikzpicture}
\end{center}
The $\Om$ operators obey the non-abelian 2-cocycle condition:
\begin{equation}
    \Omega_{g,h}\Omega_{gh,k}={}^g\Omega_{h,k}\Omega_{g,hk}.\label{eq:NA2cocycle3}
\end{equation}
Graphically, this constraint is represented by gluing four 1-cubes together to form a 2-cube:
\begin{center}
    \begin{tikzpicture}
        \node at (0,0){$U_g^AU_h^AU_k^A$};
        \node at (0,4){$U_{ghk}^A$};
        \node at (-1.9,2){$U_g^AU_{hk}^A$};
        \node at (1.9,2){$U_{gh}^AU_k^A$};
        \draw[thick,black,->](-.5,3.6)--(-1.5,2.4);
        \draw[thick,black,->](.5,3.6)--(1.5,2.4);
        \draw[thick,black,->](-1.5,1.6)--(-.5,.4);
        \draw[thick,black,->](1.5,1.6)--(.5,.4);
        \node at (-1.4,3.3){$\Omega_{g,hk}$};
        \node at (1.4,3.3){$\Omega_{gh,k}$};
        \node at (-1.5,.7){${}^g\Omega_{h,k}$};
        \node at (1.5,.7){$\Omega_{g,h}$};
        \node at (0,2){$=$};
    \end{tikzpicture}
\end{center}
In this diagram, each vertex represents an operator supported near $A$. Moreover, each arrow represents an operator supported near $\partial A$, equal to $U_fU_i^{-1}$ where $U_i$ is the operator at the initial vertex and $U_f$ is the operator at the final vertex. Composition of arrows into paths corresponds to composition of operators. The two paths from the top vertex to the bottom vertex correspond to the two sides of (\ref{eq:NA2cocycle3}). Equality follows from the fact that the two paths have the same initial and final vertices.

Restricting the $\Om$ operators to the interval $I$,\footnote{We have already assumed that the symmetry has trivial $H^2(G,\mathbb{Q}_+)$ index. Let us further assume that the symmetry restrictions are chosen such that the $\Om$ operators all have trivial GNVW index.} this constraint becomes the defining formula for $\Gamma_{g.h.k}$:
\begin{equation}
    \Gamma_{g,h,k}\cdot\hp{}^g\Omega_{h,k}^I\Omega_{g,hk}^I=\Omega_{g,h}^I\Omega_{gh,k}^I.
\end{equation}
The $2$-cube of the previous diagram now represents this operator:
\begin{center}
    \begin{tikzpicture}
        \node at (0,0){$g\cdot h\cdot k$};
        \node at (0,4){$ghk$};
        \node at (-2,2){$g\cdot hk$};
        \node at (2,2){$gh\cdot k$};
        \draw[thick,black,->](-.5,3.6)--(-1.5,2.4);
        \draw[thick,black,->](.5,3.6)--(1.5,2.4);
        \draw[thick,black,->](-1.5,1.6)--(-.5,.4);
        \draw[thick,black,->](1.5,1.6)--(.5,.4);
        \node at (-1.4,3.3){$\Omega_{g,hk}^I$};
        \node at (1.4,3.3){$\Omega_{gh,k}^I$};
        \node at (-1.5,.7){${}^g\Omega_{h,k}^I$};
        \node at (1.5,.7){$\Omega_{g,h}^I$};
        \draw[thick,black,double,->](-.4,2)--(.4,2);
        \node at (0,2.5){$\Gamma_{g,h,k}$};
    \end{tikzpicture}
\end{center}
In this version of the diagram, the vertices of the 2-cube are no longer interpreted as products of restricted symmetries. Instead they are given a matching abstract label.

In addition to the non-abelian 2-cocycle condition, consider the trivial equality
\begin{equation}
    {}^{g\cdot h}\Omega_{k,l}\Omega_{g,h}=\Omega_{g,h}\cdot\hp{}^{gh}\Omega_{k,l}.
\end{equation}
This equality has the following graphical representation:
\begin{center}
    \begin{tikzpicture}
        \node at (0,0){$U^A_gU^A_hU^A_kU^A_l$};
        \node at (0,4){$U^A_{gh}U^A_{kl}$};
        \node at (-2,2){$U^A_gU^A_hU^A_{kl}$};
        \node at (2,2){$U^A_{gh}U^A_kU^A_l$};
        \draw[thick,black,->](-.5,3.6)--(-1.5,2.4);
        \draw[thick,black,->](.5,3.6)--(1.5,2.4);
        \draw[thick,black,->](-1.5,1.6)--(-.5,.4);
        \draw[thick,black,->](1.5,1.6)--(.5,.4);
        \node at (-1.4,3.3){$\Omega_{g,h}$};
        \node at (1.4,3.3){${}^{gh}\Omega_{k,l}$};
        \node at (-1.5,.7){${}^{g\cdot h}\Omega_{k,l}$};
        \node at (1.5,.7){$\Omega_{g,h}$};
        \node at (0,2){$=$};
    \end{tikzpicture}
\end{center}
Upon restricting the $\Om$ operators to interval $I$, this condition becomes the defining formula for $\Delta_{g,h,(k,l)}$:
\begin{equation}
    \Delta_{g,h,(k,l)}\cdot\hp{}^{g\cdot h}\Omega_{k,l}^I\Omega^I_{g,h}=\Omega^I_{g,h}\cdot\hp{}^{gh}\Omega^I_{k,l}.
\end{equation}
Just as before, the 2-cube of the previous diagram, with each vertex operator label replaced by an abstract label, now represents this operator:
\begin{center}
    \begin{tikzpicture}
        \node at (0,0){$g\cdot h\cdot k\cdot l$};
        \node at (0,4){$gh\cdot kl$};
        \node at (-2,2){$g\cdot h\cdot kl$};
        \node at (2,2){$gh\cdot k\cdot l$};
        \draw[thick,black,->](-.5,3.6)--(-1.5,2.4);
        \draw[thick,black,->](.5,3.6)--(1.5,2.4);
        \draw[thick,black,->](-1.5,1.6)--(-.5,.4);
        \draw[thick,black,->](1.5,1.6)--(.5,.4);
        \node at (-1.4,3.3){$\Omega_{g,h}^I$};
        \node at (1.4,3.3){${}^{gh}\Omega_{k,l}^I$};
        \node at (-1.5,.7){${}^{g\cdot h}\Omega_{k,l}^I$};
        \node at (1.5,.7){$\Omega_{g,h}^I$};
        \draw[thick,black,double,->](-.4,2)--(.4,2);
        \node at (0,2.5){$\Delta_{g,h,(k,l)}$};
    \end{tikzpicture}
\end{center}

As we saw algebraically in Sec.~\ref{sec:explicit}, the $\Ga$ and $\Delta$ operators obey the non-abelian 3-cocycle condition:

$$
\Ga_{g,h,k}\cdot\hp{}^{^g(h,k)}\Ga_{g,hk,l}\cdot\hp{}^g\Ga_{h,k,l}=\hp{}
    ^{(g,h)}\Ga_{gh,k,l}\cdot\hp{}\Delta_{g,h,(k,l)}\cdot\hp{}^{^{g\cdot h}(k,l)}\Ga_{g,h,kl}.\label{eq:NA3cocycle3}
$$
Graphically, this constraint is represented by gluing six $2$-cubes together to form a 3-cube:
\tikzmath{
\a=0; 
\b=9; 
\x=3; 
\y=1.7; 
\dy=.35; 
\dx=.7; 
\ddx=.4; 
}
\begin{center}
    \begin{tikzpicture}
        \node at (\a,4*\y){$ghkl$};
        \node at (\a-\x,3*\y){$g\cdot hkl$};
        \node at (\a+\x,3*\y){$ghk\cdot l$};
        \node at (\a-\x,\y){$g\cdot h\cdot kl$};
        \node at (\a+\x,\y){$gh\cdot k\cdot l$};
        \node at (\a,0){$g\cdot h\cdot k\cdot l$};
        \node at (\a,2*\y){$gh\cdot kl$};
        \draw[thick,black,->](\a-\dx,4*\y-\dy)--(\a-\x+\dx,3*\y+\dy);
        \draw[thick,black,->](\a+\dx,4*\y-\dy)--(\a+\x-\dx,3*\y+\dy);
        \draw[thick,black,->](\a-\x,3*\y-\dy)--(\a-\x,\y+\dy);
        \draw[thick,black,->](\a+\x,3*\y-\dy)--(\a+\x,\y+\dy);
        \draw[thick,black,->](\a-\x+\dx,\y-\dy)--(\a-\dx,\dy);
        \draw[thick,black,->](\a+\x-\dx,\y-\dy)--(\a+\dx,\dy);
        \draw[thick,black,->](\a-\dx,2*\y-\dy)--(\a-\x+\dx,\y+\dy);
        \draw[thick,black,->](\a+\dx,2*\y-\dy)--(\a+\x-\dx,\y+\dy);
        \draw[thick,black,->](\a,4*\y-\dy)--(\a,2*\y+\dy);
        \node at (\a-\x*.6,3.75*\y){$\Omega^I_{g,hkl}$};
        \node at (\a+\x*.6,3.75*\y){$\Omega^I_{ghk,l}$};
        \node at (\a-\x*1.2,2*\y){${}^g\Omega^I_{h,kl}$};
        \node at (\a+\x*1.2,2*\y){$\Omega^I_{gh,k}$};
        \node at (\a-\x*.6,.25*\y){${}^{g\cdot h}\Omega^I_{k,l}$};
        \node at (\a+\x*.6,.25*\y){$\Omega^I_{g,h}$};
        \node at (\a-\x*.55,1.8*\y){$\Omega^I_{g,h}$};
        \node at (\a+\x*.55,1.8*\y){${}^{gh}\Omega^I_{k,l}$};
        \node at (\a-\x*.2,3.25*\y){$\Omega^I_{gh,kl}$};
        \draw[double,<-,black,thick](\a+\ddx,\y)--(\a-\ddx,\y);
        \draw[double,<-,black,thick](\a-\x/2+\ddx,2.5*\y)--(\a-\x/2-\ddx,2.5*\y);
        \draw[double,<-=,black,thick](\a+\x/2+\ddx,2.5*\y)--(\a+\x/2-\ddx,2.5*\y);
        \node at (\a-.45*\x,2.5*\y+.5){$^{^{g\cdot h}(k,l)}\Gamma_{g,h,kl}$};
        \node at (\a+\x/2,2.5*\y+.5){${}^{(g,h)}\Gamma_{gh,k,l}$};
        \node at (\a,\y+.5){$\Delta_{g,h,(k,l)}$};
        \begin{scope}[shift={(-2.5,-7)}]
        \node at (\b,4*\y){$ghkl$};
        \node at (\b-\x,3*\y){$g\cdot hkl$};
        \node at (\b+\x,3*\y){$ghk\cdot l$};
        \node at (\b-\x,\y){$g\cdot h\cdot kl$};
        \node at (\b+\x,\y){$gh\cdot k\cdot l$};
        \node at (\b,0){$g\cdot h\cdot k\cdot l$};
        \node at (\b,2*\y){$g\cdot hk\cdot l$};
        \draw[thick,black,->](\b-\dx,4*\y-\dy)--(\b-\x+\dx,3*\y+\dy);
        \draw[thick,black,->](\b+\dx,4*\y-\dy)--(\b+\x-\dx,3*\y+\dy);
        \draw[thick,black,->](\b-\x,3*\y-\dy)--(\b-\x,\y+\dy);
        \draw[thick,black,->](\b+\x,3*\y-\dy)--(\b+\x,\y+\dy);
        \draw[thick,black,->](\b-\x+\dx,\y-\dy)--(\b-\dx,\dy);
        \draw[thick,black,->](\b+\x-\dx,\y-\dy)--(\b+\dx,\dy);
        \draw[thick,black,->](\b-\x+\dx,3*\y-\dy)--(\b-\dx,2*\y+\dy);
        \draw[thick,black,->](\b+\x-\dx,3*\y-\dy)--(\b+\dx,2*\y+\dy);
        \draw[thick,black,->](\b,2*\y-\dy)--(\b,\dy);
        \node at (\b-\x*.6,3.75*\y){$\Omega^I_{g,hkl}$};
        \node at (\b+\x*.6,3.75*\y){$\Omega^I_{ghk,l}$};
        \node at (\b-\x*1.2,2*\y){${}^g\Omega^I_{h,kl}$};
        \node at (\b+\x*1.2,2*\y){$\Omega^I_{gh,k}$};
        \node at (\b-\x*.6,.25*\y){${}^{g\cdot h}\Omega^I_{k,l}$};
        \node at (\b+\x*.6,.25*\y){$\Omega^I_{g,h}$};
        \node at (\b-\x*.2,\y){${}^g\Omega^I_{h,k}$};
        \node at (\b-\x*.6,2.25*\y){${}^g\Omega^I_{hk,l}$};
        \node at (\b+\x*.6,2.25*\y){$\Omega^I_{g,hk}$};
        \draw[double,<-,black,thick](\b+\ddx,3*\y)--(\b-\ddx,3*\y);
        \draw[double,<-,black,thick](\b-\x/2+\ddx,1.5*\y)--(\b-\x/2-\ddx,1.5*\y);
        \draw[double,<-,black,thick](\b+\x/2+\ddx,1.5*\y)--(\b+\x/2-\ddx,1.5*\y);
        \node at (\b-\x/2,1.5*\y+.5){${}^g\Gamma_{h,k,l}$};
        \node at (\b+\x/2,1.5*\y+.5){$\Gamma_{g,h,k}$};
        \node at (\b,3*\y+.5){${}^{{}^g(h,k)}\Gamma_{g,hk,l}$};
        \end{scope}
        \node at (3.2,-.2){\Large{=}};
    \end{tikzpicture}
\end{center}
Here, the diagram on the right represents the ``front" of the 3-cube, whereas the diagram on the left represents the ``back". In each diagram, each path of arrows from the top vertex to the bottom vertex represents an operator supported near $I$. Moreover, each 2-cube represents an operator supported near $\partial I$, which is equal to the ``difference between paths". More precisely, this operator equals $\Om_f\Om_i^{-1}$ where $\Om_i$ is the operator corresponding to the total path (\textit{i.e.}, from top to bottom) to the left of the 2-cube, and $\Om_f$ is the operator corresponding to the total path to the right of the 2-cube. 2-cubes can be composed into ``paths between paths", representing composition of the corresponding operators. In this sense, the front and back surfaces of the 3-cube represent the operators on the two sides of (\ref{eq:NA3cocycle3}). Equality follows from the fact that the two ``paths between paths" connect the same initial and final paths.

Finally, we restrict the $\Ga$ and $\Delta$ operators to point $a$ to obtain the anomaly formula (\ref{eq:omega}):
\begin{equation}
    \omega(g,h,k,l)\equiv\Ga^a_{g,h,k}\cdot\hp{}^{^g(h,k)}\Ga^a_{g,hk,l}\cdot\hp{}^g\Ga^a_{h,k,l}\left(\hp{}^{(g,h)}\Ga^a_{gh,k,l}\cdot\hp{}\Delta^a_{g,h,(k,l)}\cdot\hp{}^{^{g\cdot h}(k,l)}\Ga^a_{g,h,kl}\right)^{-1}.
    \label{eq:omega2}
\end{equation}
The 3-cube of the previous diagram now represents this $U(1)$ phase:
\begin{center}
    \begin{tikzpicture}
        \node at (\a,4*\y){$ghkl$};
        \node at (\a-\x,3*\y){$g\cdot hkl$};
        \node at (\a+\x,3*\y){$ghk\cdot l$};
        \node at (\a-\x,\y){$g\cdot h\cdot kl$};
        \node at (\a+\x,\y){$gh\cdot k\cdot l$};
        \node at (\a,0){$g\cdot h\cdot k\cdot l$};
        \node at (\a,2*\y){$gh\cdot kl$};
        \draw[thick,black,->](\a-\dx,4*\y-\dy)--(\a-\x+\dx,3*\y+\dy);
        \draw[thick,black,->](\a+\dx,4*\y-\dy)--(\a+\x-\dx,3*\y+\dy);
        \draw[thick,black,->](\a-\x,3*\y-\dy)--(\a-\x,\y+\dy);
        \draw[thick,black,->](\a+\x,3*\y-\dy)--(\a+\x,\y+\dy);
        \draw[thick,black,->](\a-\x+\dx,\y-\dy)--(\a-\dx,\dy);
        \draw[thick,black,->](\a+\x-\dx,\y-\dy)--(\a+\dx,\dy);
        \draw[thick,black,->](\a-\dx,2*\y-\dy)--(\a-\x+\dx,\y+\dy);
        \draw[thick,black,->](\a+\dx,2*\y-\dy)--(\a+\x-\dx,\y+\dy);
        \draw[thick,black,->](\a,4*\y-\dy)--(\a,2*\y+\dy);
        \draw[double,<-,black,thick](\a+\ddx,\y)--(\a-\ddx,\y);
        \draw[double,<-,black,thick](\a-\x/2+\ddx,2.5*\y)--(\a-\x/2-\ddx,2.5*\y);
        \draw[double,<-=,black,thick](\a+\x/2+\ddx,2.5*\y)--(\a+\x/2-\ddx,2.5*\y);
        \node at (\a-.45*\x,2.5*\y+.5){$^{^{g\cdot h}(k,l)}\Gamma_{g,h,kl}^a$};
        \node at (\a+\x/2,2.5*\y+.5){${}^{(g,h)}\Gamma_{gh,k,l}^a$};
        \node at (\a,\y+.5){$\Delta_{g,h,(k,l)}^a$};
        \begin{scope}[shift={(-2.5,-7)}]
        \node at (\b,4*\y){$ghkl$};
        \node at (\b-\x,3*\y){$g\cdot hkl$};
        \node at (\b+\x,3*\y){$ghk\cdot l$};
        \node at (\b-\x,\y){$g\cdot h\cdot kl$};
        \node at (\b+\x,\y){$gh\cdot k\cdot l$};
        \node at (\b,0){$g\cdot h\cdot k\cdot l$};
        \node at (\b,2*\y){$g\cdot hk\cdot l$};
        \draw[thick,black,->](\b-\dx,4*\y-\dy)--(\b-\x+\dx,3*\y+\dy);
        \draw[thick,black,->](\b+\dx,4*\y-\dy)--(\b+\x-\dx,3*\y+\dy);
        \draw[thick,black,->](\b-\x,3*\y-\dy)--(\b-\x,\y+\dy);
        \draw[thick,black,->](\b+\x,3*\y-\dy)--(\b+\x,\y+\dy);
        \draw[thick,black,->](\b-\x+\dx,\y-\dy)--(\b-\dx,\dy);
        \draw[thick,black,->](\b+\x-\dx,\y-\dy)--(\b+\dx,\dy);
        \draw[thick,black,->](\b-\x+\dx,3*\y-\dy)--(\b-\dx,2*\y+\dy);
        \draw[thick,black,->](\b+\x-\dx,3*\y-\dy)--(\b+\dx,2*\y+\dy);
        \draw[thick,black,->](\b,2*\y-\dy)--(\b,\dy);
        \draw[double,<-,black,thick](\b+\ddx,3*\y)--(\b-\ddx,3*\y);
        \draw[double,<-,black,thick](\b-\x/2+\ddx,1.5*\y)--(\b-\x/2-\ddx,1.5*\y);
        \draw[double,<-,black,thick](\b+\x/2+\ddx,1.5*\y)--(\b+\x/2-\ddx,1.5*\y);
        \node at (\b-\x/2,1.5*\y+.5){${}^g\Gamma^a_{h,k,l}$};
        \node at (\b+\x/2,1.5*\y+.5){$\Gamma^a_{g,h,k}$};
        \node at (\b,3*\y+.5){${}^{{}^g(h,k)}\Gamma^a_{g,hk,l}$};
        \end{scope}
        \node at (4.3,.4){$\omega(g,h,k,l)$};
        \begin{scope}[rotate=-35,shift={(-1.8,-1.2)}]
        \draw[thick,black](4,3)--(5,3);
        \draw[thick,black](4,3.1)--(5.1,3.1);
        \draw[thick,black](4,3.2)--(5,3.2);
        \draw[thick,black](4.9,3.3)--(5.1,3.1)--(4.9,2.9);
        \end{scope}
    \end{tikzpicture}
\end{center}
In this final diagram, the edges of the 3-cube are no longer interpreted as operators, so their labels have been erased. This completes the graphical description of our anomaly formula and its derivation.

\section {Discussion}
\label{sec:discussion}

In this work, we have used the notion of symmetry restriction to compute the anomaly of a finite group unitary symmetry acting by FDQCs on a 2D lattice system, thereby generalizing the results of \cite{ElseNayak}. Via the bulk-boundary correspondence, our results provide a method for identifying bulk 3D SPT phases.

From a mathematical perspective, this work may be regarded as a step towards a theory of finite group representations by QCAs on many-body Hilbert spaces. In this context, it is natural to regard on-site $G$-representations as trivial, and a pair of $G$-representations as equivalent if they can be transformed into one another via the following pair of operations:
\begin{enumerate}
    \item Tensoring with an on-site $G$-representation on an ancillary Hilbert space.
    \item Conjugation by an arbitrary QCA.
\end{enumerate}
A $G$-representation which is not on-site but can be transformed into an on-site symmetry via these operations is called an ``onsiteable" symmetry. The problem of classifying \textit{non}-onsiteable symmetries may be considered in any spatial dimension. In 1D, it has been established that the equivalence classes of $G$-representations are in one-to-one correspondence with anomaly classes in $H^3(G,U(1))$ \cite{SeifnashriShirley}. In other words, a 1D $G$-symmetry is onsiteable if and only if it is anomaly-free.

In Sec.~\ref{sec:index}, we defined a total index valued in $H^2(G,\mathbb{Q}_+)\times H^4(G,U(1))$, that labels equivalence classes of 2D $G$-representations. In light of this finding, there are two natural questions to ask. First, is every value of this index realized by some $G$-representation? Second, is this index complete? In other words, if two $G$-representations have the same index, do they necessarily belong to the same equivalence class? To answer the first question, we note that 2D boundary theories of 3D group cohomology SPT models \cite{ChenSPT} furnish $G$-representations with arbitrary $H^4(G,U(1))$ index \cite{ElseNayak}. Moreover, Ref.~\cite{H2} constructs examples of $G$-symmetries with arbitrary $H^2(G,\mathbb{Q}_+)$ index. By tensoring different $G$-symmetries with one another, it is therefore possible to realize $G$-representations with arbitrary total index. On the other hand, the second question remains open; we conjecture that the answer is ``yes".

It is also natural to ask how this story plays out in higher dimensions. We conjecture that, in $d$ spatial dimensions, the total obstruction to onsiteability is captured by an index valued in the degree-$(d+2)$ generalized cohomology of $G$ with coefficients in the $\Om$-spectrum of QCAs.\footnote{QCAs have been conjectured to form an $\Om$-spectrum in Ref.~\cite{LongElse}} In 1D and 2D, this index reduces to the known $H^3(G,U(1))$ and $H^2(G,\mathbb{Q}_+)\times H^4(G,U(1))$ indices. We leave a systematic study to future work.\footnote{This problem is closely related to the conjectured classification of $d$-dimensional SPTs in terms of generalized cohomology with coefficients in the $\Om$-spectrum of invertible phases \cite{Kitaev1,Kitaev2,Theo}.}

It would also be worthwhile to extend the symmetry restriction approach to antiunitary symmetries, fermionic systems, and systems with non-tensor product Hilbert spaces.\footnote{For instance, the boundary theory of the $\mathbb{Z}_2\times\mathbb{Z}_2^f$ root SPT introduced in \cite{Metlitski} has both fermionic degrees of freedom and a non-tensor product Hilbert space.}
It is also tempting to ask to what extent the method can be applied to anomalous generalized symmetries \cite{Gaiotto,McGreevy,Cordova} and crystalline symmetries \cite{Crystalline1,Crystalline2}.\newline

\textit{Note added:} During the preparation of this manuscript, we became aware of similar results obtained by Kapustin \cite{Kapustin} and Czajka, Geiko, and Thorngren \cite{Thorngren}.

\section {Acknowledgments}

We wish to thank to David Penneys and Sean Sanford for interesting discussions. We are especially grateful to Michael Levin for inspiring discussions and collaboration on a related project. The work of WS is supported by the Leinweber Institute for Theoretical Physics and the
Ultra-Quantum Matter Simons Collaboration (Simons Foundation Grant No. 651444). KK was supported by NSF DMS 1654159, the
Center for Emergent Materials which is an NSF-funded MRSEC under Grant No. DMR-2011876, and the Center for Quantum Information Science and Engineering at The Ohio State University. The authors of this paper were ordered alphabetically.

\appendix

\section{Review: Group cohomology}
\label{app:cohomology}

In this appendix, we briefly review the definitions of the cohomology groups $H^2(G,\mathbb{Q}_+)$, $H^3(G,U(1))$, and $H^4(G,U(1))$.

\subsection{$H^2(G,\mathbb{Q}_+)$}

A $\mathbb{Q}_+$-valued \textit{2-cocycle} is a function $\nu:G\times G\to\mathbb{Q}_+$ that satisfies the \textit{2-cocycle condition}
\begin{equation}
    \nu(g,h)\nu(gh,k)=\nu(h,k)\nu(g,hk)\label{eq:2cocycle}
\end{equation}
for any triple of group elements $g,h,k\in G$. The set of 2-cocycles is subject to the following equivalence relation: $\nu\sim\nu'$ if there exists a function $\mu:G\to\mathbb{Q}_+$ such that 
\begin{equation}
    \nu'(g,h)=\nu(g,h)\frac{\mu(g)\mu(h)}{\mu(gh)}.\label{eq:2coboundary}
\end{equation}
The cohomology group $H^2(G,\mathbb{Q}_+)$ is defined as the set of equivalence classes $[\nu]$ of 2-cocycles, with the composition law $[\nu]+[\nu']=[\nu\cdot\nu']$. The ratio to the right of $\nu(g,h)$ in (\ref{eq:2coboundary}) is referred to as the \textit{2-coboundary} $d\mu$.

\subsection{$H^3(G,U(1))$}

A $U(1)$-valued \textit{3-cocycle} is a function $\alpha:G\times G\times G\to U(1)$ that satisfies the \textit{3-cocycle condition}
\begin{equation}
    \alpha(g,h,k)\alpha(g,hk,l)\alpha(h,k,l)=\alpha(gh,k,l)\alpha(g,h,kl)\label{eq:3cocycle}
\end{equation}
for any 4-tuple of group elements $g,h,k,l\in G$. The set of 3-cocycles is subject to the following equivalence relation: $\alpha\sim\alpha'$ if there exists a function $\beta:G\times G\to U(1)$ such that
\begin{equation}
    \alpha'(g,h,k)=\alpha(g,h,k)\frac{\beta(g,h)\beta(gh,k)}{\beta(h,k)\beta(g,hk)}.\label{eq:3coboundary}
\end{equation}
The cohomology group $H^3(G,U(1))$ is defined as the set of equivalence classes $[\alpha]$ of 3-cocycles, with the composition law $[\alpha]+[\alpha']=[\alpha\cdot\alpha']$. The ratio to the right of $\alpha(g,h,k)$ in (\ref{eq:3coboundary}) is referred to as the \textit{3-coboundary} $d\beta$.

\subsection{$H^4(G,U(1))$}

A $U(1)$-valued \textit{4-cocycle} is a function $\omega:G\times G\times G\times G\to U(1)$ that satisfies the \textit{4-cocycle condition}
\begin{equation}
    \omega(h,k,l,m)\omega(g,hk,l,m)\omega(g,h,k,lm)=\omega(gh,k,l,m)\omega(g,h,kl,m)\omega(g,h,k,l)\label{eq:4cocycle}
\end{equation}
for any 5-tuple of group elements $g,h,k,l,m\in G$. The set of 4-cocycles is subject to the following equivalence relation: $\omega\sim\omega'$ if there exists a function $\rho:G\times G\times G\to U(1)$ such that
\begin{equation}
    \omega'(g,h,k,l)=\omega(g,h,k,l)\frac{\rho(g,h,k)\rho(g,hk,l)\rho(h,k,l)}{\rho(gh,k,l)\rho(g,h,kl)}.\label{eq:4coboundary}
\end{equation}
The cohomology group $H^4(G,U(1))$ is defined as the set of equivalence classes $[\omega]$ of 4-cocycles, with the composition law $[\omega]+[\omega']=[\omega\cdot\omega']$. The ratio to the right of $\omega(g,h,k,l)$ in (\ref{eq:4coboundary}) is referred to as the \textit{4-coboundary} $d\rho$.

\section{Review: Else-Nayak index}

In this appendix, we review the Else-Nayak \cite{ElseNayak} procedure for computing the anomaly of a unitary symmetry of a 1D spin chain. The input to the procedure is a collection $\{U_g\}$ of FDQCs satisfying the group law $U_gU_h\propto U_{gh}$, and the output is a 3-cocycle $\alpha:G\times G\times G\to U(1)$. The procedure consists of the following steps:

\begin{enumerate}[leftmargin=*]
    \item Choose a finite interval $I=[a,b]$, which is much longer than the range of the symmetry.

    \item Choose a restriction $U_g^I$ of each symmetry operator $U_g$ to the interval $I$.

    \item The restricted symmetry operators obey the group multiplication law up to composition by a 0D FDQC supported near points $a$ and $b$. In particular, for each pair $g,h\in G$, define the unitary
    \begin{equation}
        \Om_{g,h}\equiv U^I_gU^I_h\left(U^I_{gh}\right)^{-1}.
    \end{equation}
    These operators obey the ``non-abelian 2-cocycle condition" for all $g,h,k\in G$:
    \begin{equation}
        \Om_{g,h}\Om_{gh,k}=\hp{}^g\Om_{h,k}\Om_{g,hk},\label{eq:ElseNayak2cocycle}
    \end{equation}
    where $^gW$ is a shorthand for $U_g^IW(U^I_g)^{-1}$. This constraint is derived from associativity of the restricted symmetries: on one hand,
    \begin{align}
         (U^I_gU^I_h)U^I_k&=\Om_{g,h} U^I_{gh} U^I_k=\Om_{g,h}\Om_{gh,k} U^I_{ghk}.\label{eq:ghk1}
    \end{align}
    On the other hand,
    \begin{align}
        U^I_g(U^I_hU^I_k)&=U^I_g\Om_{h,k}U^I_{hk}=\hp{}^g\Om_{h,k}\Om_{g,hk}U^I_{ghk}.\label{eq:ghk2}
    \end{align}
    Comparing (\ref{eq:ghk1}) with (\ref{eq:ghk2}) begets (\ref{eq:ElseNayak2cocycle}).

    \item Choose a restriction $\Om^a_{g,h}$ of $\Om_{g,h}$ to the vicinity of point $a$. This choice is unique up to multiplication by a $U(1)$ phase.

    \item Since the $\Om$ operators satisfy the non-abelian 2-cocycle condition (\ref{eq:ElseNayak2cocycle}), the corresponding $\Om^a$ operators must obey an analogous constraint, up to a $U(1)$ phase. The function $\alpha:G\times G\times G\to U(1)$ is defined in terms of this phase. Specifically,
    \begin{equation}
    \alpha(g,h,k)=\Om^a_{g,h}\Om^a_{gh,k}\left(\hp{}^g\Om^a_{h,k}\Om^a_{g,hk}\right)^{-1},
    \end{equation}
    It was shown in \cite{ElseNayak} that $\alpha$ defined in this manner automatically obeys the 3-cocycle condition (\ref{eq:3cocycle}).
\end{enumerate}

Importantly, the restriction of $\Om_{g,h}$ to $\Om^a_{g,h}$ has an inherent phase ambiguity: we could instead restrict to $\tilde\Om^a_{g,h}=\beta(g,h)\Om^a_{g,h}$ for an arbitrary function $\beta:G\times G\to U(1)$. Given this choice, instead of $\alpha(g,h,k)$, we would obtain the 3-cocycle
\begin{equation}
    \tilde\alpha(g,h,k)=\alpha(g,h,k)\frac{\beta(g,h)\beta(gh,k)}{\beta(h,k)\beta(g,hk)}.
\end{equation}
Since $\tilde\alpha$ differs from $\alpha$ by the 3-coboundary $d\beta$, the two cocycles are equivalent. Moreover, since $\beta$ may be any 2-cochain, $\tilde\alpha$ may be any 3-cocycle equivalent to $\alpha$. It was further shown in \cite{ElseNayak} that different choices of symmetry restriction always yield equivalent 3-cocycles. Therefore, the procedure gives a well-defined cohomology class $[\alpha]\in H^3(G,U(1))$.

\section{Proof of the 4-cocycle condition}
\label{app:cocycle}

In this appendix, we prove that the function $\omega:G\times G\times G\times G\to U(1)$ defined in (\ref{eq:omega}) satisfies the 4-cocycle condition (\ref{eq:4cocycle}). For convenience we will adopt the following notation:
\begin{align}
    \bar\Om_{g,h}&\equiv\Om^I_{g,h},\\
    \bar\Ga_{g,h,k}&\equiv\Ga^a_{g,h,k},\\
    \bar\Delta_{g,h,(k,l)}&\equiv\Delta^a_{g,h,(k,l)}.
\end{align}

To prove (\ref{eq:4cocycle}), we will make use of a pair of equalities, (\ref{eq:equality1}) and (\ref{eq:equality2}). First,
\begin{equation}
    \bar\Delta_{g,h,(k,l)}\cdot\hp{}^{^{g\cdot h}(k,l)}\bar\Delta_{g,h,(kl,m)}\cdot\hp{}^{g\cdot h}\bar{\Ga}_{k,l,m}=\hp{}^{(g,h)\cdot gh}\bar{\Ga}_{k,l,m}\cdot\hp{}\bar\Delta_{g,h,^k(l,m)}\cdot\hp{}^{^{g\cdot h\cdot k}(l,m)}\bar\Delta_{g,h,(k,lm)}\label{eq:equality1}
\end{equation}
where
\begin{equation}
    \bar\Delta_{g,h,\hp{}^k(l,m)}\equiv\hp{}^{(g,h)\cdot gh\cdot k}\Om^L_{l,m}\left(\hp{}^{g\cdot h\cdot k}\Om^L_{l,m}\right)^{-1}.
\end{equation}
To demonstrate (\ref{eq:equality1}), let us first define the operator $\Ga^c_{g,h,k}$ with support near point $c$, via the following equation:
\begin{equation}
    \Ga^a_{g,h,k}\Ga^c_{g,h,k}=\Om^L_{g,h}\Om^L_{gh,k}\big(\hp{}^g\Om^L_{h,k}\Om^L_{g,hk}\big)^{-1}.
\end{equation}
Then observe that
\begin{align}
    &\bar\Delta_{g,h,(k,l)}\cdot\hp{}^{^{g\cdot h}(k,l)}\bar\Delta_{g,h,(kl,m)}\cdot\hp{}^{g\cdot h}\bar{\Ga}_{k,l,m}\\&\hspace{10pt}=\hp{}^{(g,h)\cdot gh}\big(\Om^L_{k,l}\Om^L_{kl,m}\big)\cdot\hp{}^{g\cdot h}\big(\Om^L_{k,l}\Om^L_{kl,m}\big)^{-1}\cdot\hp{}^{g\cdot h}\bar{\Ga}_{k,l,m},\label{eq:eq11}
\end{align}
whereas
\begin{align}
    &\hp{}^{(g,h)\cdot gh}\bar{\Ga}_{k,l,m}\cdot\hp{}\bar\Delta_{g,h,^k(l,m)}\cdot\hp{}^{^{g\cdot h\cdot k}(l,m)}\bar\Delta_{g,h,(k,lm)}\\&\hspace{10pt}=\hp{}^{(g,h)\cdot gh}\big(\bar\Ga_{k,l,m}\cdot\hp{}^k\Om^L_{l,m}\Om^L_{k,lm}\big)\cdot\hp{}^{g\cdot h}\big(\hp{}^k\Om^L_{l,m}\Om^L_{k,lm}\big)^{-1}\\&\hspace{10pt}=\hp{}^{(g,h)\cdot gh}\big(\Ga^a_{k,l,m}\Ga^c_{k,l,m}\cdot\hp{}^k\Om^L_{l,m}\Om^L_{k,lm}\big)\cdot\hp{}^{g\cdot h}\big(\Ga^c_{k,l,m}\cdot\hp{}^k\Om^L_{l,m}\Om^L_{k,lm}\big)^{-1}\\&\hspace{10pt}=\hp{}^{(g,h)\cdot gh}\big(\Om^L_{k,l}\Om^L_{kl,m}\big)\cdot\hp{}^{g\cdot h}\big(\Om^L_{k,l}\Om^L_{kl,m}\big)^{-1}\cdot\hp{}^{g\cdot h}\bar{\Ga}_{k,l,m}.\label{eq:eq12}
\end{align}
Comparing (\ref{eq:eq11}) with (\ref{eq:eq12}), we obtain (\ref{eq:equality1}). The second equality is as follows:
\begin{equation}
    \bar{\Ga}_{g,h,k}\cdot\hp{}^{^g(h,k)}\bar\Delta_{g,hk,(l,m)}\cdot\hp{}^g\bar\Delta_{h,k,(l,m)}=\hp{}^{(g,h)}\bar\Delta_{gh,k,(l,m)}\cdot\hp{}\bar\Delta_{g,h,^k(l,m)}\cdot\hp{}^{^{g\cdot h\cdot k}(l,m)}\bar{\Ga}_{g,h,k}.\label{eq:equality2}
\end{equation}
To derive this equality, observe that
\begin{align}
    &\bar{\Ga}_{g,h,k}\cdot\hp{}^{^g(h,k)}\bar\Delta_{g,hk,(l,m)}\cdot\hp{}^g\bar\Delta_{h,k,(l,m)}\\
    &\hspace{10pt}=\bar{\Ga}_{g,h,k}\cdot\hp{}^{^g(h,k)\cdot(g,hk)\cdot ghk}{\Om}^L_{l,m}\cdot\hp{}(^{g\cdot h\cdot k}{\Om}^L_{l,m})^{-1}\\
    &\hspace{10pt}=\hp{}^{\Ga_{g,h,k}\cdot\hp{}^g(h,k)\cdot(g,hk)\cdot ghk}{\Om}^L_{l,m}\cdot\bar{\Ga}_{g,h,k}\cdot\hp{}(^{g\cdot h\cdot k}{\Om}^L_{l,m})^{-1}\\
    &\hspace{10pt}=\hp{}^{(g,h)\cdot(gh,k)\cdot ghk}{\Om}^L_{l,m}\cdot\bar{\Ga}_{g,h,k}\cdot\hp{}(^{g\cdot h\cdot k}{\Om}^L_{l,m})^{-1}.\label{eq:eq21}
\end{align}
On the other hand,
\begin{align}
    &^{(g,h)}\bar\Delta_{gh,k,(l,m)}\cdot\hp{}\bar\Delta_{g,h,^k(l,m)}\cdot\hp{}^{^{g\cdot h\cdot k}(l,m)}\bar{\Ga}_{g,h,k}\\
    &\hspace{10pt}=\hp{}^{(g,h)\cdot(gh,k)\cdot ghk}{\Om}^L_{l,m}(^{g\cdot h\cdot k}{\Om}^L_{l,m})^{-1}\cdot\hp{}^{^{g\cdot h\cdot k}(l,m)}\bar{\Ga}_{g,h,k}\\
    &\hspace{10pt}=\hp{}^{(g,h)\cdot(gh,k)\cdot ghk}{\Om}^L_{l,m}\cdot\bar{\Ga}_{g,h,k}\cdot(^{g\cdot h\cdot k}{\Om}^L_{l,m})^{-1}.\label{eq:eq22}
\end{align}
Comparing (\ref{eq:eq21}) with (\ref{eq:eq22}), we obtain (\ref{eq:equality2}).

Finally, to prove the 4-cocycle condition, we will evaluate the expression 
\begin{equation}
    \bar{\Ga}_{g,h,k}\cdot\hp{}^{^g(h,k)}\bar{\Ga}_{g,hk,l}\cdot\hp{}^g\bar{\Ga}_{h,k,l}\cdot\hp{}^{^{g\cdot h}(k,l)\cdot\hp{}^g(h,kl)}\bar{\Ga}_{g,hkl,m}\cdot\hp{}^{g\cdot^h(k,l)}\bar{\Ga}_{h,kl,m}\cdot\hp{}^{g\cdot h}\bar{\Ga}_{k,l,m}    
\end{equation}
in two different ways. For brevity we denote $\omega(g,h,k,l)$ by $\omega_{g,h,k,l}$. On one hand,
\begin{align}
    &\bar{\Ga}_{g,h,k}\cdot\hp{}^{^g(h,k)}\bar{\Ga}_{g,hk,l}\cdot\hp{}^g\bar{\Ga}_{h,k,l}\cdot\hp{}^{^{g\cdot h}(k,l)\cdot\hp{}^g(h,kl)}\bar{\Ga}_{g,hkl,m}\cdot\hp{}^{g\cdot^h(k,l)}\bar{\Ga}_{h,kl,m}\cdot\hp{}^{g\cdot h}\bar{\Ga}_{k,l,m}\\
    &\hspace{10pt}=\omega_{g,h,k,l}\cdot\hp{}^{(g,h)}\bar{\Ga}_{gh,k,l}\cdot\hp{}\bar\Delta_{g,h,(k,l)}\cdot\hp{}^{^{g\cdot h}(k,l)}\bar{\Ga}_{g,h,kl}\\\nonumber
    &\hspace{20pt}\cdot\hp{}^{^{g\cdot h}(k,l)\cdot\hp{}^g(h,kl)}\bar{\Ga}_{g,hkl,m}\cdot\hp{}^{g\cdot^h(k,l)}\bar{\Ga}_{h,kl,m}\cdot\hp{}^{g\cdot h}\bar{\Ga}_{k,l,m}\\
    &\hspace{10pt}=\omega_{g,h,k,l}\cdot\omega_{g,h,kl,m}\cdot\hp{}^{(g,h)}\bar{\Ga}_{gh,k,l}\cdot\hp{}^{(g,h)\cdot\hp{}^{gh}(k,l)}\bar{\Ga}_{gh,kl,m}\cdot\hp{}\bar\Delta_{g,h,(k,l)}\\\nonumber
    &\hspace{20pt}\cdot\hp{}^{^{g\cdot h}(k,l)}\bar\Delta_{g,h,(kl,m)}\cdot\hp{}^{g\cdot h}\bar{\Ga}_{k,l,m}\cdot\hp{}^{^{g\cdot h\cdot k}(l,m)\cdot\hp{}^{g\cdot h}(k,lm)}\bar{\Ga}_{g,h,klm}\\
    &\hspace{10pt}=\omega_{g,h,k,l}\cdot\omega_{g,h,kl,m}\cdot\hp{}^{(g,h)}\bar{\Ga}_{gh,k,l}\cdot\hp{}^{(g,h)\cdot\hp{}^{gh}(k,l)}\bar{\Ga}_{gh,kl,m}\cdot\hp{}^{(g,h)\cdot gh}\bar{\Ga}_{k,l,m}\label{eq:equality3}\\\nonumber
    &\hspace{20pt}\cdot\hp{}\bar\Delta_{g,h,^k(l,m)}\cdot\hp{}^{^{g\cdot h\cdot k}(l,m)}\bar\Delta_{g,h,(k,lm)}\cdot\hp{}^{^{g\cdot h\cdot k}(l,m)\cdot\hp{}^{g\cdot h}(k,lm)}\bar{\Ga}_{g,h,klm}\\
    &\hspace{10pt}=\omega_{g,h,k,l}\cdot\omega_{g,h,kl,m}\cdot\omega_{gh,k,l,m}\cdot\hp{}^{(g,h)\cdot(gh,k)}\bar{\Ga}_{ghk,l,m}\cdot\hp{}^{(g,h)}\bar\Delta_{gh,k,(l,m)}\cdot\hp{}\bar\Delta_{g,h,^k(l,m)}\label{eq:4cocycle1}\\\nonumber
    &\hspace{20pt}\cdot\hp{}^{^{g\cdot h\cdot k}(l,m)}\left(^{(g,h)}\bar{\Ga}_{gh,k,lm}\cdot\hp{}\bar{\Delta}_{g,h,(k,lm)}\cdot\hp{}^{^{g\cdot h}(k,lm)}\bar{\Ga}_{g,h,klm}\right).
\end{align}
To obtain (\ref{eq:equality3}) we have used (\ref{eq:equality1}). On the other hand,
\begin{align}
    &\bar{\Ga}_{g,h,k}\cdot\hp{}^{^g(h,k)}\bar{\Ga}_{g,hk,l}\cdot\hp{}^g\bar{\Ga}_{h,k,l}\cdot\hp{}^{^{g\cdot h}(k,l)\cdot\hp{}^g(h,kl)}\bar{\Ga}_{g,hkl,m}\cdot\hp{}^{g\cdot^h(k,l)}\bar{\Ga}_{h,kl,m}\cdot\hp{}^{g\cdot h}\bar{\Ga}_{k,l,m}\\
    &\hspace{10pt}=\omega_{h,k,l,m}\cdot\hp{}\bar{\Ga}_{g,h,k}\cdot\hp{}^{^g(h,k)}\bar{\Ga}_{g,hk,l}\cdot\hp{}^{^g(h,k)\cdot\hp{}^g(hk,l)}\bar{\Ga}_{g,hkl,m}\\\nonumber
    &\hspace{20pt}\cdot\hp{}^g\left(\hp{}^{(h,k)}\bar{\Ga}_{hk,l,m}\cdot\hp{}\bar\Delta_{h,k,(l,m)}\cdot\hp{}^{^{h\cdot k}(l,m)}\bar{\Ga}_{h,k,lm}\right)\\
    &\hspace{10pt}=\omega_{h,k,l,m}\cdot\omega_{g,hk,l,m}\cdot\hp{}^{(g,h)\cdot(gh,k)}\bar{\Ga}_{ghk,l,m}\bar{\Ga}_{g,h,k}\cdot\hp{}^{^g(h,k)}\bar\Delta_{g,hk,(l,m)}\\\nonumber
    &\hspace{20pt}\cdot\hp{}^g\bar\Delta_{h,k,(l,m)}\cdot\hp{}^{^{g\cdot h\cdot k}(l,m)\cdot\hp{}^g(h,k)}\bar{\Ga}_{g,hk,lm}\cdot\hp{}^{g\cdot\hp{}^{h\cdot k}(l,m)}\bar{\Ga}_{h,k,lm}\\
    &\hspace{10pt}=\omega_{h,k,l,m}\cdot\omega_{g,hk,l,m}\cdot\hp{}^{(g,h)\cdot(gh,k)}\bar{\Ga}_{ghk,l,m}\cdot\hp{}^{(g,h)}\bar\Delta_{gh,k,(l,m)}\cdot\hp{}\bar\Delta_{g,h,^k(l,m)}\label{eq:equality4}\\\nonumber
    &\hspace{20pt}\cdot\hp{}^{^{g\cdot h\cdot k}(l,m)}\left(\bar{\Ga}_{g,h,k}\cdot\hp{}^{^g(h,k)}\bar{\Ga}_{g,hk,lm}\cdot\hp{}^{g}\bar{\Ga}_{h,k,lm}\right)\\
    &\hspace{10pt}=\omega_{h,k,l,m}\cdot\omega_{g,hk,l,m}\cdot\omega_{g,h,k,lm}\cdot\hp{}^{(g,h)\cdot(gh,k)}\bar{\Ga}_{ghk,l,m}\cdot\hp{}^{(g,h)}\bar\Delta_{gh,k,(l,m)}\cdot\hp{}\bar\Delta_{g,h,^k(l,m)}\label{eq:4cocycle2}\\\nonumber
    &\hspace{20pt}\cdot\hp{}^{^{g\cdot h\cdot k}(l,m)}\left(^{(g,h)}\bar{\Ga}_{gh,k,lm}\cdot\hp{}\bar{\Delta}_{g,h,(k,lm)}\cdot\hp{}^{^{g\cdot h}(k,lm)}\bar{\Ga}_{g,h,klm}\right).
\end{align}
To obtain (\ref{eq:equality4}) we have used (\ref{eq:equality2}). Comparing (\ref{eq:4cocycle1}) with (\ref{eq:4cocycle2}), we obtain the 4-cocycle condition (\ref{eq:4cocycle}).

\section{Coboundary ambiguity of the anomaly index}
\label{app:ambiguity}

In this appendix, we demonstrate that different choices of restricted operators $\Om^I_{g,h}$ and $U^A_g$ give rise to equivalent to 4-cocycles.

\subsection{Ambiguity from $\Om^I_{g,h}$}
\label{appendix:ambiguity1}

First, we consider the choice of restriction $\Om^I_{g,h}$ of $\Om'_{g,h}$ to $I$ in step 5 of the procedure. Suppose we instead choose the restriction
\begin{equation}
    \wt\Om^I_{g,h}=\La_{g,h}\Om^I_{g,h},
\end{equation}
where $\Lambda_{g,h}=\Lambda^a_{g,h}\Lambda^b_{g,h}$ is an arbitrary unitary supported in the vicinity of points $a$ and $b$. Given this choice, instead of $\Ga_{g,h,k}$ we obtain the operator
\begin{align}
    \wt\Ga_{g,h,k}&=\wt\Om_{g,h}\wt\Om_{gh,k}\left(\hp{}^g\wt\Om_{h,k}\wt\Om_{g,hk}\right)^{-1}\\&=\La_{g,h}\cdot\hp{}^{(g,h)}\La_{gh,k}\Ga_{g,h,k}\left(\hp{}^g\La_{h,k}\cdot\hp{}^{^g(h,k)}\La_{g,hk}\right)^{-1}.
\end{align}
Moreover, instead of $\Delta_{g,h,(k,l)}$, we obtain the operator
\begin{equation}
    \wt\Delta_{g,h,(k,l)}\equiv\hp{}^{\wt{(g,h)}\cdot gh}\wt\Om^I_{k,l}\left(\hp{}^{g\cdot h}\wt\Om^I_{k,l}\right)^{-1},
\end{equation}
where $\wt{(g,h)}$ is a shorthand for $\wt\Om^I_{g,h}$. We are free to decompose $\wt\Om_{g,h}^I$ as $\wt\Om^I_{g,h}=\wt\Om^L_{g,h}\wt\Om^R_{g,h}$ where $\wt\Om^L_{g,h}=\La^a_{g,h}\Om^L_{g,h}$ and $\wt\Om^R_{g,h}=\La^b_{g,h}\Om^R_{g,h}$. Thus the canonical restriction of $\wt\Delta_{g,h,(k,l)}$ can be expressed as
\begin{align}
    \wt\Delta^a_{g,h,(k,l)}&=\hp{}^{\wt{(g,h)}\cdot gh}\wt\Om^L_{k,l}\left(\hp{}^{g\cdot h}\wt\Om^L_{k,l}\right)^{-1}\\&=\La^a_{g,h}\cdot\hp{}^{(g,h)\cdot gh}\La^a_{k,l}\Delta^a_{g,h,(k,l)}\left(\hp{}^{g\cdot h}\La^a_{k,l}\cdot\hp{}^{^{g\cdot h}(k,l)}\La^a_{g,h}\right)^{-1}.
\end{align}
To continue the procedure, next we choose a restriction of $\wt\Ga_{g,h,k}$ to point $a$. We are free to choose, in particular,
\begin{equation}
    \widetilde\Ga^a_{g,h,k}=\La^a_{g,h}\cdot\hp{}^{(g,h)}\La^a_{gh,k}\Ga^a_{g,h,k}\left(\hp{}^g\La^a_{h,k}\cdot\hp{}^{^g(h,k)}\La^a_{g,hk}\right)^{-1}.
\end{equation}
Finally we obtain the 4-cocycle
\begin{equation}
    \wt\omega(g,h,k,l)=\wt\Ga^a_{g,h,k}\cdot\hp{}^{^g\wt{(h,k)}}\wt\Ga^a_{g,hk,l}\cdot\hp{}^g\wt\Ga^a_{h,k,l}\left(\hp{}^{\wt{(g,h)}}\wt\Ga^a_{gh,k,l}\cdot\hp{}\wt\Delta^a_{g,h,(k,l)}\cdot\hp{}^{^{g\cdot h}\wt{(k,l)}}\wt\Ga^a_{g,h,kl}\right)^{-1}.
\end{equation}
We now show that $\wt\omega(g,h,k,l)=\omega(g,h,k,l)$. For convenience, we use the notation $\bar\La_{g,h}\equiv\La^a_{g,h}$ and $\bar\Om_{g,h}\equiv\Om^I_{g,h}$. First, we evaluate the following expression:
\begin{align}
    &\wt\Gamma^a_{g,h,k}\cdot\hp{}^{^g\wt{(h,k)}}\wt\Gamma^a_{g,hk,l}\cdot\hp{}^g\wt\Gamma^a_{h,k,l}\\\nonumber
    &\hspace{10pt}=\bar\La_{g,h}\cdot\hp{}^{(g,h)}\bar\La_{gh,k}\wt\Ga^a_{g,h,k}\cdot\hp{}^{^g(h,k)}\bar\La^{-1}_{g,hk}\cdot\hp{}^g\bar\La^{-1}_{h,k}\\
    &\hspace{20pt}\cdot\hp{}^g\bar\La_{h,k}\cdot\hp{}^g\bar\Omega_{h,k}\bar\La_{g,hk}\cdot\hp{}^{(g,hk)}\bar\La_{ghk,l}\wt\Ga^a_{g,hk,l}\cdot\hp{}^{^g(hk,l)}\bar\La^{-1}_{g,hkl}\cdot\hp{}^g\bar\La^{-1}_{hk,l}\cdot\hp{}^g\bar\Omega^{-1}_{h,k}\cdot\hp{}^g\bar\La^{-1}_{h,k}\\\nonumber
    &\hspace{20pt}\cdot\hp{}^g\bar\La_{h,k}\cdot\hp{}^{g\cdot(h,k)}\bar\La_{hk,l}\cdot\hp{}^g\wt\Ga^a_{h,k,l}\cdot\hp{}^{g\cdot\hp{}^h(k,l)}\bar\La^{-1}_{h,kl}\cdot\hp{}^{g\cdot h}\bar\La^{-1}_{k,l}\\
    &\hspace{10pt}=\bar\La_{g,h}\cdot\hp{}^{(g,h)}\bar\La_{gh,k}\cdot\hp{}^{(g,h)\cdot(gh,k)}\bar\La_{ghk,l}\left(\wt\Ga^a_{g,h,k}\cdot\hp{}^{^g(h,k)}\wt\Ga^a_{g,hk,l}\cdot\hp{}^g\wt\Ga^a_{h,k,l}\right)\label{eq:OmegaAmbi1}\\\nonumber
    &\hspace{20pt}\cdot\hp{}^{^{g\cdot h}(k,l)\cdot\hp{}^g(h,kl)}\bar\La^{-1}_{g,hkl}\cdot\hp{}^{g\cdot\hp{}^h(k,l)}\bar\La^{-1}_{h,kl}\cdot\hp{}^{g\cdot h}\bar\La^{-1}_{k,l}.
\end{align}
Second, we evaluate the expression
\begin{align}
    &^{\wt{(g,h)}}\wt\Ga^a_{gh,k,l}\cdot\hp{}\wt\Delta^a_{g,h,(k,l)}\cdot\hp{}^{^{g\cdot h}\wt{(k,l)}}\wt\Ga^a_{g,h,kl}\\\nonumber
    &\hspace{10pt}=\bar\La_{g,h}\bar\Omega_{g,h}\bar\La_{gh,k}\cdot\hp{}^{(gh,k)}\bar\La_{ghk,l}\Ga^a_{gh,k,l}\cdot\hp{}^{^{gh}(k,l)}\bar\La^{-1}_{gh,kl}\cdot\hp{}^{gh}\bar\La^{-1}_{k,l}\bar\Omega_{g,h}^{-1}\bar\La_{g,h}^{-1}\\
    &\hspace{20pt}\cdot\bar\La_{g,h}\cdot\hp{}^{(g,h)\cdot gh}\bar\La_{k,l}\Delta^a_{g,h,(k,l)}\cdot\hp{}^{^{g\cdot h}(k,l)}\bar\La^{-1}_{g,h}\cdot\hp{}^{g\cdot h}\bar\La^{-1}_{k,l}\\\nonumber
    &\hspace{20pt}\cdot\hp{}^{g\cdot h}\bar\La_{k,l}\cdot\hp{}^{g\cdot h}\bar\Om_{k,l}\bar\La_{g,h}\cdot\hp{}^{(g,h)}\bar\La_{gh,kl}\Ga^a_{g,h,kl}\cdot\hp{}^{^{g}(h,kl)}\bar\La^{-1}_{g,hkl}\cdot\hp{}^{g}\bar\La^{-1}_{h,kl}\cdot\hp{}^{g\cdot h}\bar\Om_{k,l}^{-1}\cdot\hp{}^{g\cdot h}\bar\La_{k,l}^{-1}\\
    &\hspace{10pt}=\bar\La_{g,h}\cdot\hp{}^{(g,h)}\bar\La_{gh,k}\cdot\hp{}^{(g,h)\cdot(gh,k)}\bar\La_{ghk,l}\left(\hp{}^{(g,h)}\Ga^a_{gh,k,l}\cdot\hp{}^{(g,h)}\Delta^a_{g,h,(k,l)}\cdot\hp{}^{^{g\cdot h}(k,l)}\Ga^a_{g,h,kl}\right)\label{eq:OmegaAmbi2}\\\nonumber
    &\hspace{20pt}\cdot\hp{}^{^{g\cdot h}(k,l)\cdot\hp{}^{g}(h,kl)}\bar\La^{-1}_{g,hkl}\cdot\hp{}^{g\cdot\hp{}^h(k,l)}\bar\La^{-1}_{h,kl}\cdot\hp{}^{g\cdot h}\bar\La_{k,l}^{-1}.
\end{align}
Comparing (\ref{eq:OmegaAmbi1}) with (\ref{eq:OmegaAmbi2}), we find that $\wt\omega(g,h,k,l)=\omega(g,h,k,l)$.

\subsection{Ambiguity from $U_g^A$}
\label{appendix:ambiguity2}

Next, we consider the choice of restricted symmetries $U^A_g$. Suppose we instead choose the symmetry restrictions
\begin{equation}
    \wt{U}^A_g=\Sigma_gU^A_g
\end{equation}
where $\Sigma_g$ is an arbitrary 1D QCA supported in the vicinity of $\partial A$. Given this choice, instead of $\Omega_{g,h}$ we obtain the operator
\begin{equation}
    \wt\Om_{g,h}=\Sigma_g\cdot\hp{}^g\Sigma_h\Om_{g,h}\Sigma_{gh}^{-1}.
\end{equation}

Let us define the function $\mu:G\to\mathbb{Q}_+$ where $\mu(g)=\text{Ind}(\Sigma_g)$. Instead of the 2-cocycle $\nu(g,h)=\text{Ind}(\Om_{g,h})$, we obtain an equivalent 2-cocycle
\begin{equation}
    \wt\nu(g,h)=\nu(g,h)\frac{\mu(g)\mu(h)}{\mu(gh)}.
\end{equation}
Thus, different choices of symmetry restriction give 2-cocycles that differ by the 2-coboundary $d\mu$.

To continue the procedure, we define the operator $\wt{T}_{g,h}=T_{\wt\nu(g,h)}^{-1}$ acting on $\mathcal{H}_{\partial{A}}$.\footnote{In principle, the operator $\wt{T}_{g,h}$ may require ancilla qudits of dimension different from that of the original ancillary system. For simplicity we will assume that is not the case, and that $\wt{T}_{g,h}$ and $T_{g,h}$ act on the same Hilbert space $\mathcal{H}_{\partial A}$.} Furthermore, let us define an operator $T_g=T_{\mu(g)}^{-1}$ for each $g\in G$, such that $\Sigma_g\otimes T_g$ is a 1D FDQC. Due to (\ref{eq:Tcondition}), we may decompose $\wt T_{g,h}$ as
\begin{equation}
    \wt T_{g,h}=T_gT_hT_{g,h}T_{gh}^{-1}.
\end{equation}
Next, we must choose a restriction of $\wt\Omega'_{g,h}\equiv\wt\Om_{g,h}\otimes \wt T_{g,h}$ to the interval $I$. To make a convenient choice, we first choose a restriction $\bar\Sigma$ of $\Sigma'_g\equiv\Sigma_g\otimes T_g$ to the interval $I'=[a-\delta,b-\delta]$, where $\delta$ is some distance larger than the range of the symmetry. We then choose the following restriction of $\wt\Om_{g,h}'$:
\begin{equation}
    \wt\Om^I_{g,h}=\bar\Sigma_g\cdot\hp{}^g\bar\Sigma_h\bar\Om_{g,h}(\bar\Sigma_{gh})^{-1}
\end{equation}
where we use the notation $\bar\Om_{g,h}\equiv\Om^I_{g,h}$. Given this choice, in place of $\Ga_{g,h,k}$ we obtain the operator
\begin{align}
    \wt\Ga_{g,h,k}&=\wt\Om^I_{g,h}\wt\Om^I_{gh,k}\left(\hp{}^g\wt\Om^I_{h,k}\wt\Om^I_{g,hk}\right)^{-1}\\
    &=\bar\Sigma_g\cdot\hp{}^g\bar\Sigma_h\bar\Om_{g,h}\cdot\hp{}^{gh}\bar\Sig_k\bar\Om_{gh,k}\bar\Om_{g,hk}^{-1}\cdot\hp{}^g\bar\Sig_{hk}^{-1}\bar\Sig_g^{-1}\cdot\hp{}^{\wt{g}}(\bar\Sigma_{hk}\bar\Om_{h,k}^{-1}\cdot\hp{}^h\bar\Sig_k^{-1}\bar\Sig_h^{-1})\\
    &=\bar\Sigma_g\cdot\hp{}^g\bar\Sigma_h\cdot\hp{}^{(g,h)\cdot gh}\bar\Sig_k\Ga_{g,h,k}\cdot\hp{}^{^g\bar\Sig_{hk}^{-1}}\bar{\bar\Sig}_g\cdot\hp{}^{g\cdot h}\bar\Sig_k^{-1}\cdot\hp{}^g\bar\Sig_h^{-1}\Sig_g^{-1},
\end{align}
where $\tilde{g}$ is a shorthand for $\wt{U}^A_g$. Here we have defined the operator $\bar{\bar\Sigma}_g=
\bar\Sigma^{-1}_g\Sigma_g$. We note that since $\bar\Sig_g$ is supported on the interval $I'$, it follows that $\bar{\bar\Sig}_g$ commutes with $\bar\Om_{h,k}$.

Instead of $\Delta_{g,h,(k,l)}$, we obtain the operator
\begin{equation}
    \wt\Delta_{g,h,(k,l)}=\hp{}^{\wt{(g,h)}\cdot\wt{gh}}\wt\Om_{k,l}\cdot\hp{}^{\wt{g}\cdot\wt{h}}\wt\Om^{-1}_{k,l},
\end{equation}
where $\wt{(g,h)}$ is a shorthand for $\wt\Om^I_{g,h}$. Let us choose an arbitrary decomposition $\bar\Sig_g=\Sig^L_g\Sig^R_g$. We are free to decompose $\wt\Om_{g,h}^I$ as $\wt\Om^I_{g,h}=\wt\Om^L_{g,h}\wt\Om^R_{g,h}$ where $\wt\Om^L_{g,h}=\Sig^L_g\cdot\hp{}^g\Sig^L_h\Om^L_{g,h}(\Sig^L_{gh})^{-1}$ and similarly for $\wt\Om^R_{g,h}$. Thus the canonical restriction of $\wt\Delta_{g,h,(k,l)}$ can be expressed as
\begin{align}
    \wt\Delta^a_{g,h,(k,l)}&=\hp{}^{\wt{(g,h)}\cdot\wt{gh}}\wt\Om^L_{k,l}\left(\hp{}^{\wt{g}\cdot\wt{h}}\wt\Om^L_{k,l}\right)^{-1}\\\nonumber
    &=\bar\Sig_g\cdot\hp{}^g\bar\Sig_h\bar\Om_{g,h}\bar\Sig_{gh}^{-1}\cdot\hp{}^{\wt{gh}}\big(\Sig^L_k\cdot\hp{}^k\Sig^L_l\Om^L_{k,l}(\Sig^L_{kl})^{-1}\big)\bar\Sig_{gh}\bar\Om_{g,h}^{-1}\cdot\hp{}^g\bar\Sig_h^{-1}\bar\Sig_g^{-1}\left(\hp{}^{\wt{g}\cdot\wt{h}}\wt\Om^L_{k,l}\right)^{-1}.
\end{align}

Proceeding onward, we choose a restriction of $\wt\Ga_{g,h,k}$ to point $a$. We choose for simplicity
\begin{align}
    \wt\Ga^a_{g,h,k}&=\bar\Sigma_g\cdot\hp{}^g\bar\Sigma_h\cdot\hp{}^{(g,h)\cdot gh}\bar\Sig_k\Ga^a_{g,h,k}\cdot\hp{}^{^g\bar\Sig_{hk}^{-1}}\bar{\bar\Sig}_g\cdot\hp{}^{g\cdot h}\bar\Sig_k^{-1}\cdot\hp{}^g\bar\Sig_h^{-1}\Sig_g^{-1}.
\end{align}
Finally, we obtain the cocycle
\begin{equation}
    \wt\omega(g,h,k,l)=\wt\Ga^a_{g,h,k}\cdot\hp{}^{^{\wt{g}}\wt{(h,k)}}\wt\Ga^a_{g,hk,l}\cdot\hp{}^{\wt{g}}\wt\Ga^a_{h,k,l}\left(\hp{}^{\wt{(g,h)}}\wt\Ga^a_{gh,k,l}\cdot\hp{}\wt\Delta^a_{g,h,(k,l)}\cdot\hp{}^{^{\wt{g}\cdot \wt{h}}\wt{(k,l)}}\wt\Ga^a_{g,h,kl}\right)^{-1}.
\end{equation}
We now show that $\wt\omega(g,h,k,l)=\omega(g,h,k,l)$. First, we evaluate the following expression:
\begin{align}
    &\wt\Ga^a_{g,h,k}\cdot\hp{}^{^{\wt{g}}\wt{(h,k)}}\wt\Ga^a_{g,hk,l}\cdot\hp{}^{\wt{g}}\wt\Ga^a_{h,k,l}\\\nonumber
    &\hspace{10pt}=\bar\Sig_g\cdot\hp{}^g\bar\Sig_h\cdot\hp{}^{(g,h)\cdot gh}\bar\Sig_k\Ga^a_{g,h,k}\cdot\hp{}^{^g\bar\Sig_{hk}^{-1}}\bar{\bar\Sig}_g\cdot\hp{}^{g\cdot h}\bar\Sig_k^{-1}\cdot\hp{}^g\bar\Sig_h^{-1}\Sig_g^{-1}\\\nonumber
    &\hspace{20pt}\cdot\Sig_g\cdot\hp{}^g\bar\Sig_h\cdot\hp{}^{g\cdot h}\bar\Sig_k\cdot\hp{}^g\bar\Om_{h,k}\cdot\hp{}^g\bar\Sig_{hk}^{-1}\Sig_g^{-1}\cdot\\
    &\hspace{20pt}\cdot\bar\Sig_g\cdot\hp{}^g\bar\Sig_{hk}\cdot\hp{}^{(g,hk)\cdot ghk}\bar\Sig_l\Ga^a_{g,hk,l}\cdot\hp{}^{^g\bar\Sig_{hkl}^{-1}}\bar{\bar\Sig}_g\cdot\hp{}^{g\cdot hk}\bar\Sigma_l^{-1}\cdot\hp{}^g\bar\Sigma_{hk}^{-1}\Sig_g^{-1}\\\nonumber
    &\hspace{20pt}\cdot\Sig_g\cdot\hp{}^g\bar\Sig_{hk}\cdot\hp{}^g\bar\Om_{h,k}^{-1}\cdot\hp{}^{g\cdot h}\bar\Sig_k^{-1}\cdot\hp{}^g\bar\Sig_h^{-1}\Sig_g^{-1}\\\nonumber
    &\hspace{20pt}\cdot\Sig_g\cdot\hp{}^g\bar\Sig_h\cdot\hp{}^{g\cdot h}\bar\Sig_k\cdot\hp{}^{g\cdot(h,k)\cdot hk}\bar\Sig_l\cdot\hp{}^g\Ga^a_{h,k,l}\cdot\hp{}^{g\cdot\hp{}^h\bar\Sig_{kl}^{-1}}\bar{\bar\Sig}_h\cdot\hp{}^{g\cdot h\cdot k}\bar\Sigma_l^{-1}\cdot\hp{}^{g\cdot h}\bar\Sig_k^{-1}\cdot\hp{}^g\Sig_h^{-1}\Sig_g^{-1}\\
    &\hspace{10pt}=\bar\Sig_g\cdot\hp{}^g\bar\Sig_h\cdot\hp{}^{(g,h)\cdot gh}\bar\Sig_k\cdot\hp{}^{(g,h)\cdot(gh,k)\cdot ghk}\bar\Sig_l\left(\Ga^a_{g,h,k}\cdot\hp{}^{^g(h,k)}\Ga^a_{g,hk,l}\cdot\hp{}^g\Ga^a_{h,k,l}\right)\label{eq:OmegaAmbi3}\\\nonumber
    &\hspace{20pt}\cdot\hp{}^{^g\bar\Sig_{hkl}^{-1}}\bar{\bar\Sig}_g\cdot\hp{}^{g\cdot\hp{}^h\bar\Sig_{kl}^{-1}}\bar{\bar\Sig}_h\cdot\hp{}^{g\cdot h\cdot k}\bar\Sigma_l^{-1}\cdot\hp{}^{g\cdot h}\bar\Sig_k^{-1}\cdot\hp{}^g\Sig_h^{-1}\Sig_g^{-1}.
\end{align}
Second, we evaluate the expression
\begin{align}
    &^{\wt{(g,h)}}\wt\Ga^a_{gh,k,l}\cdot\hp{}\wt\Delta^a_{g,h,(k,l)}\cdot\hp{}^{^{\wt{g}\cdot \wt{h}}\wt{(k,l)}}\wt\Ga^a_{g,h,kl}\\\nonumber
    &\hspace{10pt}=\bar\Sig_g\cdot\hp{}^g\bar\Sig_h\bar\Om_{g,h}\cdot\hp{}^{gh}\bar\Sig_k\cdot\hp{}^{(gh,k)\cdot ghk}\bar\Sig_l\Ga^a_{gh,k,l}\cdot\hp{}^{^{gh}\bar\Sig_{kl}^{-1}}\bar{\bar\Sig}_{gh}\cdot\hp{}^{gh\cdot k}\bar\Sig_l^{-1}\cdot\hp{}^{gh}\bar\Sig_k^{-1}\bar{\bar\Sig}_{gh}\bar\Om_{g,h}^{-1}\cdot\hp{}^g\bar\Sig_h^{-1}\bar\Sig_g^{-1}\\
    &\hspace{20pt}\cdot\bar\Sig_g\cdot\hp{}^g\bar\Sig_h\bar\Om_{g,h}\bar\Sig_{gh}^{-1}\cdot\hp{}^{\wt{gh}}\big(\Sig^L_k\cdot\hp{}^k\Sig^L_l\Om^L_{k,l}(\Sig^L_{kl})^{-1}\big)\bar\Sig_{gh}\bar\Om_{g,h}^{-1}\cdot\hp{}^g\bar\Sig_h^{-1}\bar\Sig_g^{-1}\label{eq:Sig}\\\nonumber
    &\hspace{20pt}\cdot\bar\Sig_g\cdot\hp{}^g\bar\Sig_h\cdot\hp{}^{(g,h)\cdot gh}\bar\Sig_{kl}\Ga^a_{g,h,kl}\cdot\hp{}^{^g\bar\Sig_{hkl}^{-1}}\bar{\bar\Sig}_g\cdot\hp{}^{g\cdot h}\bar\Sig_{kl}^{-1}\cdot\hp{}^g\bar\Sig_h^{-1}\Sig_g^{-1}\cdot\hp{}^{\wt{g}\cdot\wt{h}}\big(\Sig^L_k\cdot\hp{}^k\Sig^L_l\Om^L_{k,l}(\Sig^L_{kl})^{-1}\big)^{-1}\\
    &\hspace{10pt}=\bar\Sig_g\cdot\hp{}^g\bar\Sig_h\cdot\hp{}^{(g,h)\cdot gh}\bar\Sig_k\cdot\hp{}^{(g,h)\cdot(gh,k)\cdot ghk}\bar\Sig_l\left(\hp{}^{(g,h)}\Ga^a_{gh,k,l}\cdot\Delta^a_{g,h,k,l}\cdot\hp{}^{^{g\cdot h}(k,l)}\Ga^a_{g,h,kl}\right)\label{eq:OmegaAmbi4}\\\nonumber
    &\hspace{20pt}\cdot\hp{}^{^g\bar\Sig_{hkl}^{-1}}\bar{\bar\Sig}_g\cdot\hp{}^{g\cdot\hp{}^h\bar\Sig_{kl}^{-1}}\bar{\bar\Sig}_h\cdot\hp{}^{g\cdot h\cdot k}\bar\Sig_l^{-1}\cdot\hp{}^{g\cdot h}\bar\Sig_k^{-1}\cdot\hp{}^g\Sig_h^{-1}\Sig_g^{-1}.
\end{align}
To obtain the first equality, we have used the fact that $^{^{\wt{g}\cdot \wt{h}}\wt{(k,l)}}\wt\Ga^a_{g,h,kl}=\hp{}^{^{\wt{g}\cdot \wt{h}}\wt\Om^L_{k,l}}\wt\Ga^a_{g,h,kl}$. Comparing (\ref{eq:OmegaAmbi3}) with (\ref{eq:OmegaAmbi4}), we find that $\wt\omega(g,h,k,l)=\omega(g,h,k,l)$.

\bibliography{bibliography.bib}

\begin{thebibliography}{10}
\providecommand{\url}[1]{\texttt{#1}}
\providecommand{\urlprefix}{URL }
\expandafter\ifx\csname urlstyle\endcsname\relax
  \providecommand{\doi}[1]{doi:\discretionary{}{}{}#1}\else
  \providecommand{\doi}{doi:\discretionary{}{}{}\begingroup
  \urlstyle{rm}\Url}\fi
\providecommand{\eprint}[2][]{\url{#2}}

\bibitem{Lieb}
E.~H. Lieb, T.~Schultz and D.~Mattis,
\newblock \emph{{Two soluble models of an antiferromagnetic chain}},
\newblock Annals Phys. \textbf{16}, 407 (1961),
\newblock \doi{10.1016/0003-4916(61)90115-4}.

\bibitem{Affleck}
I.~Affleck and E.~H. Lieb,
\newblock \emph{{A Proof of Part of Haldane's Conjecture on Spin Chains}},
\newblock Lett. Math. Phys. \textbf{12}, 57 (1986),
\newblock \doi{10.1007/BF00400304}.

\bibitem{Oshikawa}
M.~Oshikawa,
\newblock \emph{{Commensurability, Excitation Gap, and Topology in Quantum
  Many-Particle Systems on a Periodic Lattice}},
\newblock Phys. Rev. Lett. \textbf{84}(7), 1535 (2000),
\newblock \doi{10.1103/physrevlett.84.1535}.

\bibitem{Hastings}
M.~B. Hastings,
\newblock \emph{{Lieb-Schultz-Mattis in higher dimensions}},
\newblock Phys. Rev. B \textbf{69}, 104431 (2004),
\newblock \doi{10.1103/physrevb.69.104431},
\newblock \eprint{cond-mat/0305505}.

\bibitem{Ogata1}
Y.~Ogata and H.~Tasaki,
\newblock \emph{Lieb–schultz–mattis type theorems for quantum spin chains
  without continuous symmetry},
\newblock Communications in Mathematical Physics \textbf{372}(3), 951–962
  (2019),
\newblock \doi{10.1007/s00220-019-03343-5}.

\bibitem{Ogata2}
Y.~Ogata, Y.~Tachikawa and H.~Tasaki,
\newblock \emph{{General Lieb\textendash{}Schultz\textendash{}Mattis Type
  Theorems for Quantum Spin Chains}},
\newblock Commun. Math. Phys. \textbf{385}(1), 79 (2021),
\newblock \doi{10.1007/s00220-021-04116-9},
\newblock \eprint{2004.06458}.

\bibitem{KapustinSopenko}
A.~Kapustin and N.~Sopenko,
\newblock \emph{Anomalous symmetries of quantum spin chains and a
  generalization of the lieb-schultz-mattis theorem},
\newblock \doi{10.48550/arXiv.2401.02533} (2024), \eprint{2401.02533}.

\bibitem{CallanHarvey}
C.~G. Callan~Jr and J.~A. Harvey,
\newblock \emph{Anomalies and fermion zero modes on strings and domain walls},
\newblock Nuclear Physics B \textbf{250}(1-4), 427 (1985),
\newblock \doi{10.1016/0550-3213(85)90489-4}.

\bibitem{Ludwig}
S.~Ryu, J.~E. Moore and A.~W.~W. Ludwig,
\newblock \emph{Electromagnetic and gravitational responses and anomalies in
  topological insulators and superconductors},
\newblock Physical Review B \textbf{85}(4) (2012),
\newblock \doi{10.1103/physrevb.85.045104}.

\bibitem{LevinGu}
M.~Levin and Z.-C. Gu,
\newblock \emph{Braiding statistics approach to symmetry-protected topological
  phases},
\newblock Physical Review B \textbf{86}(11) (2012),
\newblock \doi{10.1103/physrevb.86.115109}.

\bibitem{ChenSPT}
X.~Chen, Z.-C. Gu, Z.-X. Liu and X.-G. Wen,
\newblock \emph{Symmetry protected topological orders and the group cohomology
  of their symmetry group},
\newblock Physical Review B \textbf{87}(15) (2013),
\newblock \doi{10.1103/physrevb.87.155114}.

\bibitem{Senthil}
A.~Vishwanath and T.~Senthil,
\newblock \emph{Physics of three-dimensional bosonic topological insulators:
  Surface-deconfined criticality and quantized magnetoelectric effect},
\newblock Physical Review X \textbf{3}(1) (2013),
\newblock \doi{10.1103/physrevx.3.011016}.

\bibitem{Wen}
X.-G. Wen,
\newblock \emph{Classifying gauge anomalies through symmetry-protected trivial
  orders and classifying gravitational anomalies through topological orders},
\newblock Physical Review D \textbf{88}(4) (2013),
\newblock \doi{10.1103/physrevd.88.045013}.

\bibitem{Chen1D}
X.~Chen, Z.-C. Gu and X.-G. Wen,
\newblock \emph{Classification of gapped symmetric phases in one-dimensional
  spin systems},
\newblock Physical Review B \textbf{83}(3) (2011),
\newblock \doi{10.1103/physrevb.83.035107}.

\bibitem{Schuch1D}
N.~Schuch, D.~Pérez-García and I.~Cirac,
\newblock \emph{Classifying quantum phases using matrix product states and
  projected entangled pair states},
\newblock Physical Review B \textbf{84}(16) (2011),
\newblock \doi{10.1103/physrevb.84.165139}.

\bibitem{Pollmann1D}
F.~Pollmann, E.~Berg, A.~M. Turner and M.~Oshikawa,
\newblock \emph{Symmetry protection of topological phases in one-dimensional
  quantum spin systems},
\newblock Physical Review B \textbf{85}(7) (2012),
\newblock \doi{10.1103/physrevb.85.075125}.

\bibitem{KapustinThorngren}
A.~Kapustin and R.~Thorngren,
\newblock \emph{Anomalous discrete symmetries in three dimensions and group
  cohomology},
\newblock Physical Review Letters \textbf{112}(23) (2014),
\newblock \doi{10.1103/physrevlett.112.231602}.

\bibitem{ElseNayak}
D.~V. Else and C.~Nayak,
\newblock \emph{Classifying symmetry-protected topological phases through the
  anomalous action of the symmetry on the edge},
\newblock Phys. Rev. B \textbf{90}, 235137 (2014),
\newblock \doi{10.1103/physrevb.90.235137}.

\bibitem{CZX}
X.~Chen, Z.-X. Liu and X.-G. Wen,
\newblock \emph{Two-dimensional symmetry-protected topological orders and their
  protected gapless edge excitations},
\newblock Physical Review B \textbf{84}(23) (2011),
\newblock \doi{10.1103/physrevb.84.235141}.

\bibitem{KawagoeLevin}
K.~Kawagoe and M.~Levin,
\newblock \emph{Anomalies in bosonic symmetry-protected topological edge
  theories: Connection to f-symbols and a method of calculation},
\newblock Physical Review B \textbf{104}(11) (2021),
\newblock \doi{10.1103/physrevb.104.115156}.

\bibitem{Sahand}
S.~Seifnashri,
\newblock \emph{Lieb-schultz-mattis anomalies as obstructions to gauging
  (non-on-site) symmetries},
\newblock SciPost Physics \textbf{16}(4) (2024),
\newblock \doi{10.21468/scipostphys.16.4.098}.

\bibitem{CFT1}
O.~M. Sule, X.~Chen and S.~Ryu,
\newblock \emph{Symmetry-protected topological phases and orbifolds:
  Generalized laughlin's argument},
\newblock Phys. Rev. B \textbf{88}, 075125 (2013),
\newblock \doi{10.1103/PhysRevB.88.075125}.

\bibitem{CFT2}
N.~Bultinck, R.~Vanhove, J.~Haegeman and F.~Verstraete,
\newblock \emph{Global anomaly detection in two-dimensional symmetry-protected
  topological phases},
\newblock Phys. Rev. Lett. \textbf{120}, 156601 (2018),
\newblock \doi{10.1103/PhysRevLett.120.156601}.

\bibitem{CFT3}
A.~Tiwari, X.~Chen, K.~Shiozaki and S.~Ryu,
\newblock \emph{Bosonic topological phases of matter: Bulk-boundary
  correspondence, symmetry protected topological invariants, and gauging},
\newblock Phys. Rev. B \textbf{97}, 245133 (2018),
\newblock \doi{10.1103/PhysRevB.97.245133}.

\bibitem{CFT4}
Y.-H. Lin and S.-H. Shao,
\newblock \emph{${\mathbb{z}}_{N}$ symmetries, anomalies, and the modular
  bootstrap},
\newblock Phys. Rev. D \textbf{103}, 125001 (2021),
\newblock \doi{10.1103/PhysRevD.103.125001}.

\bibitem{CFT5}
C.-T. Hsieh, O.~M. Sule, G.~Y. Cho, S.~Ryu and R.~G. Leigh,
\newblock \emph{Symmetry-protected topological phases, generalized laughlin
  argument, and orientifolds},
\newblock Phys. Rev. B \textbf{90}, 165134 (2014),
\newblock \doi{10.1103/PhysRevB.90.165134}.

\bibitem{Kobayashi}
R.~Kobayashi, Y.~Li, H.~Xue, P.-S. Hsin and Y.-A. Chen,
\newblock \emph{Generalized statistics on lattice},
\newblock \doi{10.48550/arXiv.2412.01886} (2025), \eprint{2412.01886}.

\bibitem{SeifnashriShirley}
S.~Seifnashri and W.~Shirley,
\newblock \emph{Disentangling anomaly-free symmetries of quantum spin chains},
\newblock \doi{10.48550/arXiv.2503.09717} (2025), \eprint{2503.09717}.

\bibitem{GNVW}
D.~Gross, V.~Nesme, H.~Vogts and R.~F. Werner,
\newblock \emph{Index theory of one dimensional quantum walks and cellular
  automata},
\newblock Communications in Mathematical Physics \textbf{310}(2), 419–454
  (2012),
\newblock \doi{10.1007/s00220-012-1423-1}.

\bibitem{Kapustin}
A.~Kapustin,
\newblock \emph{Higher symmetries and anomalies in quantum lattice systems},
\newblock \doi{10.48550/arXiv.2505.04719} (2025), \eprint{2505.04719}.

\bibitem{H2}
W.~Shirley, C.~Zhang, W.~Ji and M.~Levin,
\newblock In preparation .

\bibitem{Farrelly}
T.~Farrelly,
\newblock \emph{A review of quantum cellular automata},
\newblock Quantum \textbf{4}, 368 (2020),
\newblock \doi{10.22331/q-2020-11-30-368}.

\bibitem{FHH}
J.~Haah, L.~Fidkowski and M.~B. Hastings,
\newblock \emph{Nontrivial quantum cellular automata in higher dimensions},
\newblock Communications in Mathematical Physics \textbf{398}(1), 469–540
  (2022),
\newblock \doi{10.1007/s00220-022-04528-1}.

\bibitem{HastingsFreedman}
M.~Freedman and M.~B. Hastings,
\newblock \emph{Classification of quantum cellular automata},
\newblock Communications in Mathematical Physics \textbf{376}(2), 1171–1222
  (2020),
\newblock \doi{10.1007/s00220-020-03735-y}.

\bibitem{brown1982cohomology}
K.~S. Brown,
\newblock \emph{Cohomology of Groups}, vol.~87 of \emph{Graduate Texts in
  Mathematics},
\newblock Springer,
\newblock ISBN 978-0-387-90688-1,
\newblock \doi{10.1007/978-1-4684-9327-6} (1982).

\bibitem{Muger}
M.~Muger,
\newblock \emph{{Conformal orbifold theories and braided crossed
  G-categories}},
\newblock Commun. Math. Phys. \textbf{260}, 727 (2005),
\newblock \doi{10.1007/s00220-005-1291-z},
\newblock [Erratum: Commun.Math.Phys. 260, 763 (2005)],
\newblock \eprint{math/0403322}.

\bibitem{Ranard}
D.~Ranard, M.~Walter and F.~Witteveen,
\newblock \emph{A converse to lieb–robinson bounds in one dimension using
  index theory},
\newblock Annales Henri Poincaré \textbf{23}(11), 3905–3979 (2022),
\newblock \doi{10.1007/s00023-022-01193-x}.

\bibitem{Kitaev1}
A.~Kitaev,
\newblock \emph{Toward topological classification of phases with short-range
  entanglement},
\newblock Lecture at KITP  (2011),
\newblock \doi{https://online.kitp.ucsb.edu/online/topomat11/kitaev}.

\bibitem{Kitaev2}
A.~Kitaev,
\newblock \emph{On the classification of short-range entangled states},
\newblock Lecture at SCGP  (2013),
\newblock \doi{http://scgp.stonybrook.edu/archives/7874}.

\bibitem{LRE}
X.~Chen, Z.-C. Gu and X.-G. Wen,
\newblock \emph{Local unitary transformation, long-range quantum entanglement,
  wave function renormalization, and topological order},
\newblock Physical Review B \textbf{82}(15) (2010),
\newblock \doi{10.1103/physrevb.82.155138}.

\bibitem{KapustinYangSopenko}
A.~Kapustin, N.~Sopenko and B.~Yang,
\newblock \emph{A classification of invertible phases of bosonic quantum
  lattice systems in one dimension},
\newblock Journal of Mathematical Physics \textbf{62}(8) (2021),
\newblock \doi{10.1063/5.0055996}.

\bibitem{KitaevWire}
A.~Y. Kitaev,
\newblock \emph{Unpaired majorana fermions in quantum wires},
\newblock Physics-Uspekhi \textbf{44}(10S), 131–136 (2001),
\newblock \doi{10.1070/1063-7869/44/10s/s29}.

\bibitem{Yoshida}
B.~Yoshida,
\newblock \emph{Gapped boundaries, group cohomology and fault-tolerant logical
  gates},
\newblock Annals of Physics \textbf{377}, 387–413 (2017),
\newblock \doi{10.1016/j.aop.2016.12.014}.

\bibitem{Propitius}
M.~de~Wild~Propitius,
\newblock \emph{Topological interactions in broken gauge theories},
\newblock \doi{10.48550/arXiv.hep-th/9511195} (1995), \eprint{hep-th/9511195}.

\bibitem{WangLevin}
C.~Wang and M.~Levin,
\newblock \emph{Topological invariants for gauge theories and
  symmetry-protected topological phases},
\newblock Physical Review B \textbf{91}(16) (2015),
\newblock \doi{10.1103/physrevb.91.165119}.

\bibitem{DecoratedDomainWall}
X.~Chen, Y.-M. Lu and A.~Vishwanath,
\newblock \emph{Symmetry-protected topological phases from decorated domain
  walls},
\newblock Nature Communications \textbf{5}(1) (2014),
\newblock \doi{10.1038/ncomms4507}.

\bibitem{LongElse}
D.~M. Long and D.~V. Else,
\newblock \emph{Topological phases of many-body localized systems: Beyond
  eigenstate order},
\newblock \doi{10.48550/arXiv.2408.00825} (2024), \eprint{2408.00825}.

\bibitem{Theo}
D.~Gaiotto and T.~Johnson-Freyd,
\newblock \emph{Symmetry protected topological phases and generalized
  cohomology},
\newblock Journal of High Energy Physics \textbf{2019}(5) (2019),
\newblock \doi{10.1007/jhep05(2019)007}.

\bibitem{Metlitski}
R.~A. Jones and M.~A. Metlitski,
\newblock \emph{1d lattice models for the boundary of 2d ``majorana" fermion
  spts: Kramers-wannier duality as an exact $z_2$ symmetry},
\newblock \doi{10.48550/arXiv.1902.05957} (2019), \eprint{1902.05957}.

\bibitem{Gaiotto}
D.~Gaiotto, A.~Kapustin, N.~Seiberg and B.~Willett,
\newblock \emph{Generalized global symmetries},
\newblock Journal of High Energy Physics \textbf{2015}(2) (2015),
\newblock \doi{10.1007/jhep02(2015)172}.

\bibitem{McGreevy}
J.~McGreevy,
\newblock \emph{Generalized symmetries in condensed matter},
\newblock Annual Review of Condensed Matter Physics \textbf{14}(1), 57–82
  (2023),
\newblock \doi{10.1146/annurev-conmatphys-040721-021029}.

\bibitem{Cordova}
C.~Cordova, T.~T. Dumitrescu, K.~Intriligator and S.-H. Shao,
\newblock \emph{Snowmass white paper: Generalized symmetries in quantum field
  theory and beyond},
\newblock \doi{10.48550/arXiv.2205.09545} (2022), \eprint{2205.09545}.

\bibitem{Crystalline1}
E.~Khalaf, H.~C. Po, A.~Vishwanath and H.~Watanabe,
\newblock \emph{Symmetry indicators and anomalous surface states of topological
  crystalline insulators},
\newblock Phys. Rev. X \textbf{8}, 031070 (2018),
\newblock \doi{10.1103/physrevx.8.031070}.

\bibitem{Crystalline2}
H.~Song, S.-J. Huang, L.~Fu and M.~Hermele,
\newblock \emph{Topological phases protected by point group symmetry},
\newblock Physical Review X \textbf{7}(1) (2017),
\newblock \doi{10.1103/physrevx.7.011020}.

\bibitem{Thorngren}
A.~M. Czajka, R.~Geiko and R.~Thorngren,
\newblock In preparation .

\end{thebibliography}

\end{document}